\pdfoutput=1
\documentclass[sigconf, nonacm]{acmart}
\usepackage{mathtools}
\usepackage{enumitem}\usepackage{multirow}
\usepackage{url}
\usepackage{tikz,pgfplots}
\usetikzlibrary{positioning,shapes,arrows,shadows,patterns}
\usepackage{xcolor}
\usepackage{comment}
\usepackage{algorithm}
\usepackage[htt]{hyphenat}
\usepackage{multicol}
\usepackage[noend]{algpseudocode}
\usepackage{graphicx} 
\usepackage{amsmath, amsthm} 
\usepackage{subcaption} 
\usepackage[english]{babel}
\usepackage[utf8]{inputenc}
\newcommand\tab[1][1cm]{\hspace*{#1}}
\newcommand{\ignore}[1]{}
\newcommand{\cons}[1]{\mbox{\sc {Cons}}(#1)}
\newcommand{\cptree}{$C^+_\textit{tree}$}

\newcommand{\glb}{glb}
\newcommand{\lub}{lub}

\newcommand{\sat}{{\sc SAT}}

\newcommand{\partialmaxsat}{Partial MaxSAT}
\newcommand{\weightedmaxsat}{Weighted MaxSAT}
\newcommand{\weightedpartialmaxsat}{Weighted Partial MaxSAT}
\newcommand{\pmaxsat}{PMaxSAT}
\newcommand{\wmaxsat}{WMaxSAT}
\newcommand{\wpmaxsat}{WPMaxSAT}
\newcommand{\maxtwosat}{Max2SAT}
\newcommand{\mintwosat}{Min2SAT}

\newcommand{\pminsat}{PMinSAT}
\newcommand{\wminsat}{WMinSAT}
\newcommand{\wpminsat}{WPMinSAT}

\newcommand{\certainty}[1]{\mbox{\sc {Certainty}}(#1)}

\newtheorem{reduction}{Reduction}[section]
\newtheorem{construction}{Construction}[section]
\newcommand\vldbdoi{XX.XX/XXX.XX}
\newcommand\vldbpages{XXX-XXX}
\newcommand\vldbvolume{14}
\newcommand\vldbissue{1}
\newcommand\vldbyear{2020}
\newcommand\vldbauthors{\authors}
\newcommand\vldbtitle{\shorttitle} 
\begin{document}
\title{Consistent Answers of Aggregation Queries using SAT Solvers}
\author{Akhil A. Dixit}
\orcid{0000-0003-2138-1319}
\affiliation{\institution{University of California Santa Cruz}}
\email{akadixit@ucsc.edu}

\author{Phokion G. Kolaitis}
\affiliation{\institution{University of California Santa Cruz and IBM Research}} 
\email{kolaitis@ucsc.edu}

\begin{abstract}
The framework of database repairs and consistent answers to queries is a principled approach to managing inconsistent databases. We describe the first system able to compute the consistent answers of general aggregation queries with the \texttt{COUNT($A$)}, \texttt{COUNT(*)}, \texttt{SUM($A$)}, \texttt{MIN($A$)}, and \texttt{MAX($A$)} operators, and with or without grouping constructs. Our system uses reductions to optimization versions of Boolean satisfiability (SAT) and then leverages powerful SAT solvers. We carry out an extensive set of experiments on both synthetic and real-world data that demonstrate the usefulness and scalability of this approach.
\end{abstract}

\maketitle

\ignore{
\begingroup\small\noindent\raggedright\textbf{PVLDB Reference Format:}\\
\vldbauthors. \vldbtitle. PVLDB, \vldbvolume(\vldbissue): \vldbpages, \vldbyear.\\
\href{https://doi.org/\vldbdoi}{doi:\vldbdoi}
\endgroup
\begingroup
\renewcommand\thefootnote{}\footnote{\noindent
This work is licensed under the Creative Commons BY-NC-ND 4.0 International License. Visit \url{https://creativecommons.org/licenses/by-nc-nd/4.0/} to view a copy of this license. For any use beyond those covered by this license, obtain permission by emailing \href{mailto:info@vldb.org}{info@vldb.org}. Copyright is held by the owner/author(s). Publication rights licensed to the VLDB Endowment. \\
\raggedright Proceedings of the VLDB Endowment, Vol. \vldbvolume, No. \vldbissue\ %
ISSN 2150-8097. \\
\href{https://doi.org/\vldbdoi}{doi:\vldbdoi} \\
}\addtocounter{footnote}{-1}\endgroup
}
\section{Introduction}

The framework of database repairs and consistent query answering, introduced by Arenas, Bertossi, and Chomicki \cite{Arenas99}, is a principled approach to managing inconsistent databases, i.e., databases that violate one or more integrity constraints on their schema.
 In this framework, inconsistencies are handled at query time by considering all possible \textit{repairs} of the inconsistent database, where a repair of an inconsistent database $\mathcal{I}$ is a consistent database $\mathcal{J}$ that differs from $\mathcal{I}$ in a ``minimal'' way. The \emph{consistent answers} to a query $q$ on a given database $\mathcal{I}$ is the intersection of the results of $q$ applied on each repair of $\mathcal{I}$.
 Thus, a consistent answer provides the guarantee that it will be found no matter on what repair the query has been evaluated.
 Computing the consistent answers  can be an intractable problem, because an inconsistent database may have exponentially many repairs.
 \ignore{The main algorithmic problem in this framework is to compute the \emph{consistent answers} to a query $q$ on a given database $\mathcal{I}$, that is, the tuples that lie in the intersection of the results of $q$ applied on each repair of $\mathcal{I}$ (see the monograph \cite{Bertossi11}). Computing the consistent answers  can be an intractable problem, because an inconsistent database may have exponentially many repairs.}
 In particular, computing the consistent answers of a fixed Select-Project-Join (SPJ) query can be a coNP-complete problem. By now, there is an extensive body of work on the complexity of consistent answers for SPJ queries (see Section \ref{sec:cqa-complexity}).

 \smallskip
 
\emph{Range Semantics: Concept and Motivation.}
 Aggregation queries are the most frequently asked queries; they are of the form

\smallskip
\centerline{$Q :=\; \texttt{SELECT } Z, f(A)  \texttt{ FROM } T(U, Z, A) \texttt{ GROUP BY }Z,$}
\smallskip
\noindent where $f(A)$ is one 
 the standard aggregation operators
 \texttt{COUNT($A$)}, \texttt{COUNT(*)},
 \texttt{SUM($A$)},   \texttt{AVG($A$)}.
 \texttt{MIN($A$)},  \texttt{MAX($A$)},
and $T(U,Z,A)$ is  the relation returned by a SPJ query $q$ expressed in SQL. 
A \emph{scalar aggregation} query is an aggregation query without a \texttt{GROUP BY} clause.

What is the semantics of an aggregation query over an inconsistent database?
Since an aggregation query may return  different answers on different repairs of an inconsistent database, there is typically  \emph{no} consistent answer as per the earlier definition of consistent answers. To obtain meaningful semantics to aggregation queries, Arenas et al. \cite{ArenasB03} introduced  the \textit{range consistent answers}.

\begin{table*}[t]
\caption{Running example -- an inconsistent database instance $\mathcal{I}$ (primary key attributes are underlined)} \label{toydb}
  \begin{tabular}{lll}
\begin{tabular}{c|c c c c}
\texttt{CUSTOMER}&\underline{CID}&CNAME&CITY\\\cline{2-4}
&C1&John&LA&$f_1$\\
&C2&Mary&LA&$f_2$\\
&C2&Mary&SF&$f_3$\\
&C3&Don&SF&$f_4$\\
&C4&Jen&LA&$f_5$\\
\end{tabular}
&
\begin{tabular}{c|c c c c c}
\texttt{ACCOUNTS}&\underline{ACCID}&TYPE&CITY&BAL\\\cline{2-5}
&A1&Checking&LA&900&$f_6$\\
&A2&Checking&LA&1000&$f_7$\\
&A3&Saving&SJ&1200&$f_8$\\
&A3&Saving&SF&-100&$f_9$\\
&A4&Saving&SJ&300&$f_{10}$
\end{tabular}
&
\begin{tabular}{c|c c c}
\texttt{CUSTACC}&CID&ACCID\\\cline{2-3}
&C1&A1&$f_{11}$\\
&C2&A2&$f_{12}$\\
&C2&A3&$f_{13}$\\
&C3&A4&$f_{14}$\\
\end{tabular}
\end{tabular}
 \label{table:run}
\end{table*}

Let $Q$ be a scalar aggregation query and let $\Sigma$ be a set of integrity constraints. The set of \textit{possible answers} to $Q$ on an inconsistent   instance $\mathcal{I}$ w.r.t.\  $\Sigma$   is the set of the answers to $Q$ over all  repairs of $\mathcal{I}$ w.r.t.\ $\Sigma$, i.e.,  $\text{Poss}(Q, \Sigma) = \{Q(\mathcal{J})\; |\; \mathcal{J} \text{ is a repair of }\mathcal{I} \text{ w.r.t.\ }\Sigma\}$. By definition, the \emph{range consistent answers}  to $Q$ on $\mathcal{I}$ is the interval $[\glb(Q,\mathcal{I}),
\lub(Q,\mathcal{I})]$, where the endpoints of this interval are, respectively, the greatest lower bound (glb) and the least upper bound (lub) of the set $\text{Poss}(Q, \Sigma)$
of possible answers to $Q$ on $\mathcal I$. For example, the range consistent answers to the query

\smallskip{\small 
\noindent\texttt{SELECT SUM(ACCOUNTS.BAL) FROM ACCOUNTS, CUSTACC\\WHERE ACCOUNTS.ACCID = CUSTACC.ACCID AND CUSTACC.CID = `C2'}}
\smallskip

\ignore{
\smallskip\\
\texttt{\tab[1cm]SELECT MAX(ACCOUNTS.BAL) FROM ACCOUNTS, CUSTACC\\\tab[1cm]WHERE ACCOUNTS.ACCID = CUSTACC.ACCID\\\tab[1.5cm]AND CUSTACC.CID = `C2'}
\smallskip\\
}
\noindent on the instance in Table \ref{table:run} is the interval $[900, 2200]$. The meaning is that no matter how the database $\mathcal I$ is repaired, the answer to the query is guaranteed to be in the range between 900 and  2200.
 
Arenas et al. \cite{Arenas03} focused on scalar aggregation queries only. Fuxman, Fazli, and Miller \cite{Fuxman05}  extended the notion of range consistent answers  to aggregation queries with grouping (see Section \ref{sec-aggregation}).

\ignore{
 For a query

\smallskip
\centerline{$Q :=\; \texttt{SELECT } Z, f(A)  \texttt{ FROM } T(U, Z, A) \texttt{ GROUP BY }Z,$}
\smallskip
\noindent a tuple $(T, [glb, lub])$ is a \emph{range  consistent answer} to $Q$ on $\mathcal{I}$ if:\\

\begin{itemize}
\item For every repair $\mathcal{J}$ of $\mathcal{I}$, there exists $d$ s.t.\ $(T, d) \in Q(\mathcal{J})$ and $glb \leq d \leq lub$.
\item For some repair $\mathcal{J}$ of $\mathcal{I}$, we have  $(T, glb) \in Q(\mathcal{J})$
\item For some repair $\mathcal{J}$ of $\mathcal{I}$, we have $(T, lub) \in Q(\mathcal{J})$.
\end{itemize}

(i) For every repair $\mathcal{J}$ of $\mathcal{I}$, there exists $d$ such that $(T, d) \in Q(\mathcal{J})$ and $glb \leq d \leq lub$; (ii)  For some repair $\mathcal{J}$ of $\mathcal{I}$, we have  $(T, glb) \in Q(\mathcal{J})$; (iii) 
 For some repair $\mathcal{J}$ of $\mathcal{I}$, we have  $(T, lub) \in Q(\mathcal{J})$.
 }

Range semantics have become the standard semantics of aggregation queries in the framework of database repairs (see \cite[Section 5.6]{Bertossi11}). Furthermore, range semantics have been adapted to give semantics to aggregation queries in several other contexts, including data exchange \cite{DBLP:conf/pods/AfratiK08} and ontologies \cite{DBLP:journals/ws/KostylevR15}. Finally, range semantics have  been suggested as an alternative way to overcome some of the issues arising from SQL's handling of null values \cite{DBLP:conf/pods/GuagliardoL16}.


\smallskip

\emph{Earlier Systems for Consistent Query Answering.}
Several academic prototype systems for consistent query answering have been developed \cite{Arenas03,Barcelo03,ChomickiH04,Fuxman05,FuxmanM05,Greco03,Kolaitis13,MannaRT11,MarileoB10,DixitK19}. These systems use different approaches, including   logic programming \cite{Barcelo03,Greco03}, compact representations of repairs \cite{Chomicki04}, or reductions to solvers \cite{MannaRT11,Kolaitis13,DixitK19}. In particular, in \cite{DixitK19}, we reported on CAvSAT, a system that at that time was able to compute the consistent answers of unions of SPJ queries w.r.t.\ denial constraints (which include functional dependencies as a special case) via reductions to SAT solvers.
Among all these systems, however, only  the ConQuer system by Fuxman et al.\ \cite{Fuxman05,FuxmanM05} is capable of handling aggregation queries. Actually,  ConQuer can only handle a restricted class of aggregation query, namely, those aggregation queries w.r.t.\ key constraints for which  the underlying SPJ query belongs to the class called $C_\textit{forest}$. For such a query $Q$, the range consistent answers of $Q$ are SQL-rewritable, which means that there is a SQL query $Q'$ such that the range semantics answers of $Q$ on an instance  $\mathcal I$ can be obtained by directly evaluating $Q'$ on $\mathcal I$. This leaves out, however, many aggregation queries, including all aggregation queries whose range consistent answers   are not SQL-rewritable or are NP-hard to compute. Up to now, no system supports such queries.




\smallskip

\emph{Summary of Contributions.}
In this paper, we report on and evaluate the performance of 
AggCAvSAT (Aggregate Consistent Answers via Satisfiability Testing), which is 
an enhanced version of CAvSAT and is also  
 the first system that is capable to compute  the range consistent answers to all aggregation queries involving the operators \texttt{SUM($A$)}, \texttt{COUNT($A$)}, or \texttt{COUNT(*)}  with or without grouping.

We first corroborate the need for  a system that goes well beyond ConQuer by showing that there is an aggregation query $Q$ involving \texttt{SUM($A$)} such that the consistent answers of the underlying SPJ query $q$ w.r.t.\ key constraints are SQL-rewritable, but the range consistent answers of $Q$ are NP-hard (Theorem \ref{theorem1} in Section \ref{sec-aggregation}).

The distinctive feature of AggCAvSAT is that it uses polynomial-time reductions to reduce  the range consistent answers of aggregation queries to  optimization variants of Boolean Satisfiability (SAT),  such as Partial MaxSAT and Weighted Partial MaxSAT. These reductions, described in Sections \ref{main-body} and \ref{beyond-key-constraints}, are natural but are much more sophisticated than the reductions used in \cite{DixitK19} to reduce the consistent answers of SPJ queries to SAT. After the reductions have been carried out, AggCAvSAT
 deploys powerful SAT solvers, such as the MaxHS solver \cite{Davies2011}, to compute the range consistent answers of aggregation queries. Furthermore, AggCAvSAT can handle databases that are inconsistent not only w.r.t.\ key constraints, but also w.r.t.\ arbitrary \emph{denial} constraints,  a much broader class of constraints.
 
 An extensive experimental evaluation of AggCAvSAT is reported in Section \ref{sec:experiments}. We carried out a suite of experiments on both synthetic and real-word databases, and for a variety of aggregation queries with and without grouping. The synthetic databases were generated using two different methods: (a) the \texttt{DBGen} tool of TPC-H was used to generate consistent data and then inconsistencies were injected artificially; (b)   the PDBench inconsistent database generator from the probabilistic database management system MayBMS \cite{Antova08} was used. The experiments demonstrated the scalability of AggCAvSAT along both the size of the data and the degree of inconsistency in the data. 
 Note that AggCAvSAT was also competitive in comparison to ConQuer (especially when the degree of inconsistency was not excessive), even though the latter is tailored to only handle a restricted class of aggregation queries whose range consistent answers are SQL-rewritable.
 

\smallskip

\emph{Consistent Answers vs.\ Data Cleaning.}
There is a large body of work on managing inconsistent databases via data cleaning. There are fundamental differences between the framework of the consistent answers and the framework of data cleaning (see \cite[Section 6]{Bertossi11}). In particular, the  consistent answers provide the guarantee that each such answer will be found no matter on which repair  the query at hand is evaluated, while data cleaning provides no similar  guarantee. Data cleaning has the attraction that it produces a single consistent instance but the process need not be deterministic and the  instance  produced need not even be a repair (i.e., it need not be a maximally consistent instance).
Recent data cleaning systems, such as 
 HoloClean 
 \cite{Rekatsinas17} and Daisy
 \cite{DBLP:conf/icde/GiannakopoulouK20,DBLP:conf/sigmod/GiannakopoulouK20},
 produce 
  a probabilistic database instance as the output  (which again need not be a repair).

It is not clear how  to compare query answers over the database returned by a data cleaning system and the (range) consistent answers computed by a  consistent query answering system. 
In fact, no such comparison is given in the HoloClean 
 \cite{Rekatsinas17} and Daisy
 \cite{DBLP:conf/icde/GiannakopoulouK20,DBLP:conf/sigmod/GiannakopoulouK20} papers.
At the performance level, the data cleaning approaches  remove inconsistencies in the data offline, hence the time-consuming tasks are done prior to answering the queries; in contrast,  systems for consistent query answering work online. 
 It is an interesting project, left for future research, to develop a methodology and carry out a fair comparison on a level playing field between 
 systems for data
 cleaning and systems
 for consistent query answering.
 
\ignore{
Finally, it should be noted that there is a large body of work on managing inconsistent databases via data cleaning. Two recent systems in this area are HoloClean 
 \cite{Rekatsinas17} and Daisy
 \cite{DBLP:conf/icde/GiannakopoulouK20,DBLP:conf/sigmod/GiannakopoulouK20}.
 There are fundamental differences between data cleaning systems and systems for consistent query answering.
 Specifically, data cleaning systems, given an inconsistent database as input, produce either a single database instance or a probabilistic database instance as  output. This instance need not be a repair of the original inconsistent instance.  As a result, it is hard to compare query answers over the database cleaned by HoloClean or Daisy on the one hand and the consistent answers computed by AggCAvSAT on the other hand. Furthermore, at the performance level, the data cleaning approaches tend to remove inconsistencies in the data offline, hence the time-consuming tasks are done prior to answering the queries; in contrast,  systems for consistent query answering work online. 
 It is an interesting project, left for future research, to develop a methodology and carry out a fair comparison between 
 systems for data
 cleaning  and systems
 for consistent query answering. 
 }

\section{Preliminaries}
\subsubsection*{Integrity Constraints and Database Queries}

A relational database schema $\mathcal{R}$ is a finite collection of relation symbols, each with a fixed positive integer as its arity. The attributes of a relation symbol are names for its columns; they can be identified with their positions, thus  $Attr(R) = \{1, ..., n\}$ denotes the set of  attributes of $R$. An  $\mathcal{R}$-\emph{instance}  is a collection $\mathcal{I}$ of finite relations $R^\mathcal{I}$, one for each relation symbol $R$ in $\mathcal R$. An expression of the form $R^\mathcal{I} (a_1, ..., a_n)$ is a \textit{fact} of the instance $\mathcal{I}$ if $(a_1, ..., a_n) \in R^\mathcal{I}$.
\ignore{
Relational database schemata are typically accompanied by a set of \textit{integrity constraints}, i.e.,   rules that impose semantic restrictions on the allowable instances.}
\ignore{A \textit{functional dependency} $X \rightarrow Y$ on a relation $R$ 
asserts that if two tuples in $R$ agree on the attributes in $X$, then they must also agree on the attributes in $Y$.} A 
\textit{key} is a minimal subset $X$ of $Attr(R)$ such that the functional dependency $X \rightarrow Attr(R)$ holds.

Starting with Codd's seminal work \cite{Codd70,Codd72}, first-order logic has been successfully used as a database query language; in fact, it forms the core of SQL. A \textit{conjunctive query} is expressible by a first-order formula of the form $q({\bf z}):= \exists {\bf w}\; (R_1({\bf x}_1) \land ... \land R_m({\bf x}_m))$, where each ${\bf x}_i$ is a tuple of variables and constants, $\bf z$ and $\bf w$ are tuples of variables with no variable in common, and the variables in ${\bf x}_1, \cdots, {\bf x}_m$ appear in exactly one of the tuples $\bf z$ and $\bf w$. A conjunctive query with no free variables (i.e., all variables are existentially quantified) is a boolean query, while a conjunctive query with $k$ free variables in $\bf z$ is a $k$-ary query. Conjunctive queries are also known as \emph{select-project-join} (SPJ) queries with equi-joins, and are among the most frequently asked queries. For example, on the instance  $\mathcal{I}$ from Table \ref{toydb}, the binary conjunctive query $q(z,x):= \exists w\; (\texttt{CUST}(w,x,y) \land \texttt{CUSTACC}(w,z))$ returns the set of all pairs $(z,x)$ such that $z$ is an account ID of an account owned by customer named $x$.

Equivalently, this query can be expressed in SQL as
\smallskip\\
\texttt{\tab[1cm]SELECT CUSTACC.ACCID, CUST.CNAME\\\tab[1cm]FROM CUST, CUSTACC\\\tab[1cm]WHERE CUST.CID = CUSTACC.CID}
\smallskip\\
A  \textit{union of conjunctive queries} is expressible by a disjunction $q({\bf z}) := q_1 \vee \cdots \vee q_n$ of conjunctive queries,  where all conjunctive queries $q_i$ have  the same arity. Unions of conjunctive queries are strictly more expressive than conjunctive queries.

\subsubsection*{Database Repairs and Consistent Answers} Let $\Sigma$ be a set  of integrity constraints  on a database schema $\mathcal{R}$.
An $\mathcal R$-instance $\mathcal{I}$ is \emph{consistent} if $\mathcal{I} \models \Sigma$, i.e., $\mathcal{I}$ satisfies every constraint in $\Sigma$; otherwise, $\mathcal{I}$ is  \emph{inconsistent}. 
For example, let $\mathcal I$ be the instance depicted in Table \ref{table:run}. There are two key constraints, namely, \texttt{CUST}(CID) and \texttt{ACC}(ACCID). Clearly,  $\mathcal{I}$ is inconsistent since the facts $f_2, f_3$ of \texttt{CUST} and facts $f_8, f_9$ of \texttt{ACC} violate  these key constraints.

A \emph{repair} of an inconsistent  instance $\mathcal{I}$ w.r.t.\ $\Sigma$  is a  consistent instance $\mathcal{J}$ that differs from $\mathcal{I}$ in a ``minimal" way.  Different notions of minimality give rise to different types of repairs (see \cite{Bertossi11} for a comprehensive survey). Here, we  focus on \emph{subset repairs}, the most extensively studied type of repairs. An instance $\mathcal{J}$ is  a \emph{subset repair} of an  instance $\mathcal{I}$ if
$\mathcal J$ is a maximal consistent subinstance of $\mathcal I$, that is, 
$\mathcal{J} \subseteq \mathcal{I}$ (where $\mathcal{I}$ and $\mathcal{J}$ are viewed as sets of facts), $\mathcal{J}\models \Sigma$, and there exists no instance $\mathcal{J'}$ such that $\mathcal{J'}\models \Sigma$ and $\mathcal{J} \subset \mathcal{J'}\subset \mathcal{I}$. Arenas et al.\ \cite{Arenas99} used repairs to give rigorous semantics to query answering on inconsistent databases. Specifically, assume that $q$ is a query,  $\mathcal{I}$ is an $\mathcal R$-instance, and $\bf t$ is a tuple of values. We  say that $\bf t$ is a \emph{consistent answer} to $q$ on $\mathcal{I}$ w.r.t.\ $\Sigma$ if ${\bf t} \in q(\mathcal{J})$, for every repair $\mathcal{J}$ of $\mathcal{I}$. We write $\cons{q, \mathcal{I}, \Sigma}$ to denote the set of all \emph{consistent answers} to $q$ on  $\mathcal{I}$ w.r.t.\ $\Sigma$, i.e.,

\centerline{

$\cons{q, \mathcal{I}, \Sigma} = \bigcap\; \{q(\mathcal{J}): \mbox{$\mathcal{J}$ is a repair of $\mathcal{I}$ w.r.t.\ $\Sigma$}\}.$}

If $\Sigma$ is a fixed set of integrity constraints and $q$ is a fixed  query, then the main computational problem associated with the consistent answers is: given an instance $\mathcal{I}$,  compute \cons{$q$, $\mathcal{I}$, $\Sigma$}; we write $\cons{q,\Sigma}$ to denote this problem.
If $q$ is a boolean  query, then computing the consistent answers becomes the decision problem $\certainty{q,\Sigma}$: given an instance $\mathcal{I}$,
is $q$ true on every repair $\mathcal{J}$ of $\mathcal{I}$ w.r.t.\ $\Sigma$? When the constraints in $\Sigma$ are understood from the context, we will write 
$\cons{q}$ and
$\certainty{q}$ in place of $\cons{q,\Sigma}$
and 
$\certainty{q,\Sigma}$, respectively.


\subsubsection*{Complexity of Consistent Answers} \label{sec:cqa-complexity}
There has been an extensive study of the consistent answers of conjunctive queries \cite{Bertossi11,FuxmanM05,Fuxman07,Wijsen09,WijsenJ10,WijsenR10W,Kolaitis12,Koutris16,Koutris17}. If $\Sigma$ is a fixed set of 
key constraints and $q$ 
is a boolean conjunctive query, then  $\certainty{q,\Sigma}$ is always in coNP, but, depending on the query and the constraints,  $\certainty{q,\Sigma}$ exhibits a variety of behaviors within coNP.
The most definitive result to date is a \emph{trichotomy} theorem by Koutris and Wijsen \cite{Koutris16,Koutris17}; it asserts that if $q$ is a self-join-free (no repeated relation symbols) boolean conjunctive query with one key constraint per relation, then $\certainty{q}$ is either SQL-rewritable, or in P but not SQL-rewritable, or coNP-complete. 
\ignore{
Moreover, there is a quadratic algorithm to decide, given such a query, which of the three cases of the trichotomy holds.} It is an open problem whether or not this trichotomy extends to arbitrary boolean conjunctive queries and to broader classes of constraints (e.g., denial constraints).

We illustrate the trichotomy theorem with three examples.
\ignore{for which the complexity of CQA  was pinpointed before the trichotomy theorem was established, thus providing a hint for the general result.} In what follows,  the underlined attributes constitute the keys to the relations.  First, if $q_1$ is the  query $\exists x, y,z(R(\underline{x},z)\land S(\underline{z},y))$, then $\certainty{q_1}$ is SQL-rewrtiable \cite{Fuxman07}. Second, if $q_2$ is
the query $\exists x,y (R(\underline{x},y)\land S(\underline{y},x))$, then $\certainty{q_2}$ is in P, but is not SQL-rewritable \cite{WijsenR10W}. Third, if $q_3$ is the query $\exists x,y,z (R(\underline{x},y)\land S(\underline{z},y))$, then
$\certainty{q_3}$ is coNP-complete \cite{Fuxman07}.

\subsubsection*{Boolean Satisfiability and SAT Solvers}
Boolean Satisfiability (\sat{}) is arguably the prototypical and  the most widely studied NP-complete problem. \sat{} is the following decision problem:  \textit{given a boolean formula $\varphi$, is $\varphi$ satisfiable?} 
There has been an extensive body of research on different aspects of boolean satisfiability 
(see the handbook  \cite{Biere09}). In particular, significant progress has been made on developing 
\sat{}-solvers, so much so that the advances in this area of research are often referred to as the  ``SAT Revolution'' \cite{Vardi09}).
Typically, a \sat{}-solver takes a boolean formula $\varphi$ in \textit{conjunctive normal form} (CNF) as an input and outputs a satisfying assignment for $\varphi$ (if one exists) or tells that the formula $\varphi$ is unsatisfiable. Recall that a formula $\varphi$ is in CNF if it is a conjunction of clauses, where each clause is a disjunction of literals. For example, the formula $ (x_1 \lor x_2 \lor \neg x_3) \land (\neg x_2 \lor x_3) \land (\neg x_1 \lor x_4)$ has a satisfying assignment $x_1 = 1$, $x_2 = 0$, $x_3 = 0$, and $x_4 = 1$.

At present, \sat{}-solvers are capable of solving quickly \sat{}-instances with millions of clauses and variables. \sat{}-solvers have been widely used in both academia and industry as general-purpose tools. Indeed, many real-world problems from a variety of domains, including scheduling, protocol design, software verification, and model checking, can be naturally encoded as \sat{}-instances, and solved quickly using solvers,  such as Glucose \cite{AudemardS09} and CaDiCaL \cite{cadical}. Furthermore, \sat{}-solvers have  been used in solving open problems in mathematics \cite{DBLP:conf/sat/HeuleKM16,DBLP:conf/lpar/OostemaMH20}.
In \cite{DixitK19}, we used \sat{}-solvers  to build a prototypical system for consistent query answering, which we called CAvSAT. This system can compute consistent answers to unions of conjunctive queries over relational databases that are inconsistent w.r.t.\ a fixed set of arbitrary denial constraints.
\section{
Range Consistent Answers}\label{sec-aggregation}
Frequently asked database queries often involve one of the standard aggregation operators   \texttt{COUNT($A$)}, \texttt{COUNT(*)}, \texttt{SUM($A$)}, \texttt{AVG($A$)}, \texttt{MIN($A$)},  \texttt{MAX($A$)}, and, possibly,  
a \texttt{GROUP BY} clause. In what follows, we will use the term
\emph{aggregation queries} to refer to queries with aggregate operators and with or without a \texttt{GROUP BY} clause. Thus, in full generality, an aggregation query can be expressed as
\smallskip\\
\centerline{$Q :=\; \texttt{SELECT } Z, f(A)  \texttt{ FROM } T(U, Z, A) \texttt{ GROUP BY }Z,$}\smallskip\\
where $f(A)$ is one of the aforementioned aggregate operators and $T(U,Z,A)$ is the relation returned by a query $q$, which typically is a conjunctive query or a union of conjunctive queries expressed in SQL. This  query $q$  
 is called the \textit{underlying  query} 
 of $Q$, the attribute represented by the variable $w$ is called the
   \textit{aggregation attribute}, and the attributes represented by $Z$ are called the \textit{grouping attributes}.
A \emph{scalar aggregation} query is one without a \texttt{GROUP BY} clause.

\ignore{
$T(U, Z, w)$ is a relation expressible using a union of conjunctive queries, $f$ is one of \texttt{MIN($w$)}, \texttt{MAX($w$)}, \texttt{SUM($w$)}, \texttt{COUNT($w$)}, \texttt{COUNT(*)} or \texttt{AVG($w$)}.
on top of a conjunctive query or a union of conjunctive queries. We refer to these as aggregation queries, and they can be expressed as\smallskip\\
\centerline{$Q :=\; \texttt{SELECT }Z, f \texttt{ FROM } T(U, Z, w) \texttt{ GROUP BY }Z,$}\smallskip\\
where $T(U, Z, w)$ is a relation expressible using a union of conjunctive queries, $f$ is one of \texttt{MIN($w$)}, \texttt{MAX($w$)}, \texttt{SUM($w$)}, \texttt{COUNT($w$)}, \texttt{COUNT(*)} or \texttt{AVG($w$)}. A conjunctive query (or a union of conjunctive queries) that expresses relation $T(U, Z, w)$ is called an \textit{underlying conjunctive query} (or an \textit{underlying union of conjunctive queries}) of $Q$. The attribute represented by the variable $w$ is referred as the \textit{aggregation attribute}, and the ones represented by $Z$ are called the \textit{grouping attributes}. The aggregation attribute may be preceded by \texttt{DISTINCT} keyword in SQL, which is used to ignore duplicate values of the aggregation attribute.
}

It is often the case that an aggregation query returns  different answers on different repairs of an inconsistent database; thus, even for a scalar aggregation query, there is typically  \emph{no} consistent answer as per the  definition of consistent answers given earlier. In fact, 
to produce an empty set of consistent answers,
it suffices to have just two repairs on which a scalar aggregation query returns difference answers.  Aiming  to obtain more meaningful answers to aggregation queries, Arenas et al. \cite{ArenasB03}  proposed the \textit{range consistent answers}, as an  alternative notion of consistent answers.

Let $Q$ be a scalar aggregation query. The set of \textit{possible answers} to $Q$ on an inconsistent instance $\mathcal{I}$ consists of the answers to $Q$ over all  repairs of $\mathcal{I}$, i.e.,  $\text{Poss}(Q, \Sigma) = \{Q(\mathcal{J})\; |\; \mathcal{J} \text{ is a repair of }\mathcal{I} \text{ w.r.t.\ }\Sigma\}$. By definition, the \emph{range consistent answers}  to $Q$ on $\mathcal{I}$ is the interval $[\glb(Q,\mathcal{I}),
\lub(Q,\mathcal{I})]$, where the endpoints of this interval are, respectively, the greatest lower bound (glb) and the least upper bound (lub) of the set $\text{Poss}(Q, \Sigma)$ of possible answers to $Q$ on $\mathcal I$. 

For example,
the range consistent answers of the query
\smallskip\\{\small 
\noindent\texttt{SELECT SUM(ACCOUNTS.BAL) FROM ACCOUNTS, CUSTACC\\WHERE ACCOUNTS.ACCID = CUSTACC.ACCID AND CUSTACC.CID = `C2'}}
\smallskip

\ignore{ 
\smallskip\\
\texttt{\tab[1cm]SELECT MAX(ACC.BAL)\\\tab[1cm]FROM ACC, CUSTACC\\\tab[1cm]WHERE ACC.ACCID = CUSTACC.ACCID\\\tab[2cm]AND CUSTACC.CID = `C2'}
\smallskip\\
}

\noindent on the instance in Table \ref{table:run} is the interval $[900, 2200]$. The  guarantee is that   no matter how  the database $\mathcal I$ is repaired, the answer to the query is guaranteed to be in the range between 900 and  2200. Note that, the glb-answer comes from a repair of $\mathcal{I}$ that contains the fact $f_9$, while the lub-answer is from a repair that contains the fact $f_8$. 
 
 Arenas et al. \cite{Arenas03} focused on scalar aggregation queries only.
 Fuxman, Fazli, and Miller \cite{Fuxman05}  extended the notion of range consistent answers  to aggregation queries with grouping, i.e., to queries  
 $$Q :=\; \texttt{SELECT }Z, f(A) \texttt{ FROM } T(U, Z, A) \texttt{ GROUP BY }Z.$$ 
 For such queries, 
 a tuple $(T, [glb, lub])$ is a \emph{range  consistent answer} to $Q$ on $\mathcal{I}$, if the following conditions hold:

$\bullet$~ For every repair $\mathcal{J}$ of $\mathcal{I}$, there exists $d$ s.t.\ $(T, d) \in Q(J)$ and $glb \leq d \leq lub$.

$\bullet$~ For some repair $\mathcal{J}$ of $\mathcal{I}$, we have that $(T, glb) \in Q(J)$

$\bullet$~ For some repair $\mathcal{J}$ of $\mathcal{I}$, we have that $(T, lub) \in Q(J)$.

If $Q$ is an aggregation query,  $\cons{Q}$ denotes  the problem: given an instance $\mathcal{I}$, compute the range semantics of $Q$ on $\mathcal{I}$.
\subsubsection*{Complexity of Range Consistent Answers}
Arenas et al. \cite{ArenasB03} investigated the computational complexity of the range consistent answers for scalar aggregation queries of the form
\smallskip\\
\centerline{$\texttt{SELECT }  f(A)  \texttt{ FROM } R(U, A)$},\smallskip\\
where $f(A)$ is one of the standard aggregation operators and $R(U,A)$ is a  relational schema with  functional dependencies. The main findings in Arenas et al.~\cite{ArenasB03} can be summarized as follows.

\smallskip

$\bullet$
     If the relational schema $R(U,A)$ has at most one functional dependency and $f(A)$ is one of the aggregation operators \texttt{MIN($A$)}, \texttt{MAX($A$)}, \texttt{SUM($A$)}, \texttt{COUNT(*)},  \texttt{AVG($A$)},
    then the range consistent answers of the query $ \texttt{SELECT }  f($A$)  \texttt{ FROM } R(U, A)$ is in P.
    
\smallskip

$\bullet$
 There is a relational schema $R(U,A)$ with one key dependency such that computing the range consistent answers of the query $ \texttt{SELECT }  \texttt{COUNT}(A)  \texttt{ FROM } R(U, A)$ is an NP-complete problem.
 
 \smallskip
 
$\bullet$ There is a relational schema $R(U,A)$ with two functional dependencies, such that computing the range consistent answers of the query  $ \texttt{SELECT }  f(A)  \texttt{ FROM } R(U, A)$ is a NP-complete problem, where $f(A)$
is one of the standard aggregation operators.

\ignore{
have discovered a dichotomy in the complexity of consistent answers to scalar aggregation queries, i.e., for queries without grouping and with exactly one relation symbol, where the integrity constraints in consideration were functional dependencies. This dichotomy asserts the following. Computing consistent answers to a scalar aggregation query over a database schema with at most one functional dependency constraint is in P for aggregation functions \texttt{MIN(A)}, \texttt{MAX(A)}, \texttt{SUM(A)}, \texttt{COUNT(*)}, and \texttt{AVG(A)}. Moreover, there exists a scalar aggregation query with each of the aforementioned functions such that it is an NP-complete task to compute its consistent answers over a database schema with two functional dependency constraints. Interestingly, this problem is NP-complete for aggregation function \texttt{COUNT(A)} even for one functional dependency.}
It remains an open problem to pinpoint the complexity  of the range consistent answers for richer aggregation queries of the form 

\centerline{$Q:= \; \texttt{SELECT } Z,  f(A) \texttt{ FROM } T(U, Z, A)~ \texttt{GROUP BY}~ Z,$}
\noindent where $T(U,Z,A)$ is the relation returned by a conjunctive query $q$ or by a union $q:= q_1\cup \cdots \cup q_k$  of conjunctive queries.
It can be shown, however, that if computing the consistent answers $\cons{q}$ of the underlying query $q$ is a hard problem, then 
computing the range consistent answers $\cons{Q}$ of the aggregation query $Q$ is a hard problem as well.
This gives rise to the following question:
what can we say about the complexity of the range consistent answers $\cons{Q}$ if computing the consistent answers $\cons{q}$
of the underlying query is an easy problem?

\ignore{
Observe that computing the range semantic consistent answers to an aggregation query $Q$ with the \texttt{COUNT(*)} aggregation function is at least as hard as computing the consistent answers to the underlying conjunctive query (or a union of conjunctive queries) of $Q$, as shown in Proposition \ref{agg-hardness1}.

\begin{proposition}\label{agg-hardness1}
Let $Q := \texttt{SELECT COUNT(*) FROM } T$, where $T$ is a unary relation expressible by a boolean conjunctive query $q$. 
\begin{itemize}
    \item If $\certainty{q, \mathcal{I}} \in P$ then $\cons{Q, \mathcal{I}} \in P$.
    \item If $\certainty{q, \mathcal{I}}$ is coNP-complete, then $\cons{Q, \mathcal{I}}$ is coNP-hard.
\end{itemize}
\end{proposition}
\begin{proof}

It is easy to see that $\cons{q, \mathcal{I}, \Sigma}$ is true if and only if the \textit{lub}-answer in $\cons{Q, \mathcal{I}, \Sigma}$ is greater than 0.
\end{proof}
}

Fuxman and Miller \cite{Fuxman07} identified a class, called $C_\textit{forest}$, of self-join free conjunctive queries whose consistent answers are SQL-rewritable. In his PhD thesis, 
Fuxman \cite{FuxmanPhDthesis} 
introduced the class $C_\textit{aggforest}$ consisting of all aggregation queries such that
the aggregation operator is one of \texttt{MIN($A$)}, \texttt{MAX($A$)}, \texttt{SUM($A$)},  \texttt{COUNT(*)},
the underlying query $q$ is a conjunctive query in $C_\textit{forest}$, and   there is one key constraint for each relation in the underlying query $q$.
Fuxman \cite{FuxmanPhDthesis} showed that 
 the range consistent answers of every query
 in $C_\textit{aggforest}$
 are SQL-rewritable (earlier, similar results
for a proper subclass of $C_\textit{aggforest}$
were obtained by Fuxman, Fazli, and Miller).  

It is known that there are self-join free conjunctive queries outside the class $C_\textit{forest}$ whose
consistent answers are  SQL-rewritable. In fact, Koutris and Wijsen \cite{Koutris17} have characterized the self-join free conjunctive queries whose consistent answers are SQL rewritable.   
However, the SQL rewritability of aggregation queries beyond those in $C_\textit{aggforest}$ has not been investigated. 
 In the sequel, we show that there exists a self-join-free conjunctive query whose consistent answers are SQL-rewritable, but this property is not preserved when an aggregation operator is added on top of it. Specifically, 
 we reduce the \textsc{Maximum Cut} problem to the problem of computing the  range consistent answers to an aggregation query involving \texttt{SUM} and whose underlying conjunctive query has SQL-rewritable consistent answers. We begin by recalling the definition of the \textsc{Maximum Cut} problem, a fundamental  NP-complete problem \cite{Karp1972}. We state and prove a helping lemma (Lemma \ref{lemma1}) before stating the main result in Theorem \ref{theorem1}.
 \begin{definition}{\textsc{Maximum Cut}.}
 For an undirected graph $G = (V, E)$, a \emph{cut} of $G$ is a partition $(S, \overline{S})$ of $V$, where $S \subseteq V$ and $\overline{S} = V \backslash S$. The set of edges with one vertex in $S$ and one vertex in $\overline{S}$ is denoted by $E(S, \overline{S})$, and the the size of the cut $(S, \overline{S})$ is $|E(S, \overline{S})|$. 
 
 The \textsc{Maximum Cut} problem asks: Given an undirected graph $G$ and an integer $k$, is there a cut of $G$ that has size at least $k$? 
 \end{definition}
 
 \begin{theorem}\label{theorem1}
Let $\mathcal{R}$ be a schema with three relations $R_1(\underline{A_1}, B_1)$, $R_2(\underline{A_2}, B_2)$, and $R_3(\underline{A_1, B_1, A_2, B_2, C})$. Let $Q$ be the following aggregation query:
\begin{center}
\smallskip
$Q:= $ \texttt{SELECT SUM($A$)}\texttt{ FROM } $q(A)$,
\smallskip\\
\end{center}
\noindent where $q(A)$ is the following self-join-free conjunctive query:
\begin{center}
\smallskip
$q(A):= \exists x \exists y \; R_1(\underline{x}, \text{`red'}) \land R_2(\underline{y}, \text{`blue'}) \land R_3(\underline{x, \text{`red'}, y, \text{`blue'}, A})$.
\smallskip\\
\end{center}

\noindent Then
the following two statements hold.
\begin{enumerate}
 \item $\cons{q}$ is SQL-rewritable.
 \item $\cons{Q}$ is NP-hard.
 \end{enumerate}
\end{theorem}

\begin{proof}
 To show that
 $\cons{q}$ is SQL-rewritable,
  consider the following first-order query $q'$:
 \begin{align*}
 q'(A):= \; & \exists x \exists y ( R_1(\underline{x}, \text{`red'}) \land R_2(\underline{y}, \text{`blue'}) \land R_3(\underline{x, \text{`red'}, y, \text{`blue'}, A})\\ &\land  \forall z(R_1(x, z) \rightarrow z = \text{`red'}) \land \forall w(R_2(y, w) \rightarrow w = \text{`blue'})).
 \end{align*} 
 \ignore{It is not hard to verify that for every instance
 $\mathcal I$ and every value $a$, we have that $a\in q'(\mathcal{I}) \Leftrightarrow a \in \cons{q, \mathcal{I}}$.}
We will show that for every instance $\mathcal I$ and every value $a$, we have that $a\in q'(\mathcal{I}) \Leftrightarrow a \in \cons{q, \mathcal{I}}$. Since $q'$ filters out the tuples from $R_1$ and $R_2$ that participate in the violations of the key constraints, we have that if $a\in q'(\mathcal{I})$, then $a \in q(\mathcal{J})$, for every repair $\mathcal{J}$ of $\mathcal{I}$, which means that $a \in \cons{q, \mathcal{I}}$. In the other direction, we claim that if $a \in \cons{q, \mathcal{I}}$, then $a\in q'(\mathcal{I})$. Indeed, if $a\not \in q'(\mathcal{I})$, then for all $x$ and $y$ such that $R_1(x, \text{`red'}) \land R_2(y,\text{`blue'}) \land R_3(x,\text{`red'}, y, \text{`blue'}, a)$, we would have that there is some $z$ such that $R_1(x,z)$ and $z\not = \text{`red'}$ or there is some $w$ such that $R_2(y,w)$ and $w\not = \text{`blue'}$. Construct a repair $\mathcal J$ of $\mathcal I$ as follows. First, for every $x$, if $\text{`red'}$ is the only value $z$ such that $R_1(x,z)$ is a fact of $\mathcal I$, then put $R_1(x,\text{`red'})$ in $\mathcal J$; otherwise, pick an element $z^*\not = \text{`red'}$ such that $R_1(x,z^*)$ is a fact of $\mathcal I$ and put $R_1(x,z^*)$ in $\mathcal J$. Second, for every  $y$, if $\text{`blue'}$ is the only value $w$ such that $R_2(y,w)$ is a fact of $\mathcal I$, then put $R_2(y,\text{`blue'})$ in $\mathcal J$; otherwise, pick an element $w^*\not =  \text{`blue'}$ such that $R_1(y,w^*)$ is a fact of $\mathcal I$ and put $R_2(x,w^*)$ in $\mathcal J$. Third, put every tuple of the relation $R_3$ of $\mathcal I$ into $\mathcal J$. Clearly, $\mathcal J$ is a repair of $\mathcal I$. Moreover, $a \not \in q({\mathcal J})$. Indeed, if $a \in q({\mathcal J})$, then there are elements $x$ and $y$ such that
  ${\mathcal J}\models R_1(x, \text{`red'}) \land R_2(y,\text{`blue'}) \land R_3(x,\text{`red'}, y, \text{`blue'}, a)$. Since $a\not \in q'({\mathcal I})$, we have that there is some  $z'$ such that $R_1(x,z')$ and $z'\not = \text{`red'}$ or there is some $w'$ such that
 $R_2(y,w')$ and $w'\not = \text{`blue'}$. In the first case, the construction of $\mathcal J$ implies that $R_1(x,\text{'red'})$ is not a fact of $\mathcal J$, while in the second case, the construction of $\mathcal J$ implies that
 $R_2(y,\text{`blue'})$ is not a fact of $\mathcal J$; in either case, we have arrived at a contradiction.
 
 To show that $\cons{Q}$ is NP-hard, consider the following reduction from undirected graphs to $\mathcal R$-instances, where $\mathcal{R}$ is the  schema with relations $R_1(\underline{A_1}, B_1)$, $R_2(\underline{A_2}, B_2)$, and $R_3(\underline{A_1, B_1, A_2, B_2, C})$.
 \begin{reduction}\label{maxcut-cqa} L Given an undirected graph $G = (V, E)$, construct an $\mathcal{R}$-instance $\mathcal{I}$ as follows. Let $m = -|E| - 1$.
 \begin{flushleft}
 
\begin{itemize}[noitemsep]
\item For each $v \in V$, add tuples $R_1(v, \text{`red'})$, $R_1(v, \text{`blue'})$, $R_2(v, \text{`red'})$, and $R_2(v, \text{`blue'})$ to $\mathcal{I}$.
\item For each $v \in V$, add a tuple $R_3(v, \text{`red'}, v, \text{`blue'}, m)$ to $\mathcal{I}$.
\item For each edge $(u, v) \in E$, add tuples $R_3(u, \text{`red'}, v, \text{`blue'}, 1)$ and $R_3(v, \text{`red'}, u, \text{`blue'}, 1)$ to $\mathcal{I}$.
\end{itemize}
\end{flushleft}
\end{reduction}
We will show that the preceding Reduction \ref{maxcut-cqa} reduces \textsc{Maximum Cut} to computing the range semantics of the aggregation query $Q$.

For the rest of this section, let $G$ be an undirected graph and $\mathcal{I}$ be the database instance constructed from $G$ using Reduction \ref{maxcut-cqa}. We say that a repair $\mathcal{J'}$ of $\mathcal{I}$ produces a \emph{red-blue coloring} of $G$ if for every vertex $v \in V$, we have that the tuples $R_1(v, \text{`red'})$ and $R_2(v, \text{`red'})$ are either both present in $\mathcal{J'}$ or both absent in $\mathcal{J'}$. We now prove a useful lemma.
 
 \begin{lemma}\label{lemma1}
  For every repair $\mathcal{J}$ of $\mathcal{I}$, there exists a repair $\mathcal{J'}$ of $\mathcal{I}$ (not necessarily different from $\mathcal J$) such that $\mathcal{J'}$ produces a red-blue coloring of $G$ and $Q(\mathcal{J'}) \geq Q(\mathcal{J})$. 
\end{lemma}
\begin{proof}
Let $\mathcal{J}$ be a repair of $\mathcal{I}$. Construct an $\mathcal R$-instance $\mathcal{J'}$ from $\mathcal{J}$ as follows. For every vertex $v \in V$, if both tuples $R_1(v, x)$ and $R_2(v, x)$ are  present in $\mathcal{J}$ for $x \in \{\text{`red'}, \text{`blue'}\}$, then add them to $\mathcal{J'}$. Otherwise, add the tuples $R_1(v, \text{`red'})$ and $R_2(v, \text{`red'})$ to $\mathcal{J'}$. Also, copy all tuples from relation $R_3$ of $\mathcal{J}$ to relation $R_3$ of $\mathcal{J'}$. Clearly, $\mathcal{J'}$ is a repair of $\mathcal{I}$ and $\mathcal{J'}$ produces a red-blue coloring of $G$. Observe that $Q(\mathcal{J'})$ can be different than $Q(\mathcal{J})$ only if there exists at least one vertex $v \in V$ such that either $R_1(v, \text{`red'}), R_2(v, \text{`blue'}) \in \mathcal{J}$ or $R_1(v, \text{`blue'}), R_2(v, \text{`red'}) \in \mathcal{J}$. 

We will show that  $Q(\mathcal{J'}) \geq Q(\mathcal{J})$ holds.

\textbf{Case 1}: Let  $v$ be a node such that $R_1(v, \text{`red'}), R_2(v, \text{`blue'}) \in \mathcal{J}$. In this case, while populating the database instance $\mathcal{J'}$, vertex $v$ changes its color in relation $R_2$, i.e., we have that $R_2(v, \text{`red'}) \in \mathcal{J'}$ and $R_2(v, \text{`blue'}) \not\in \mathcal{J'}$. Therefore, the summands arising from the tuples of the form $R_1(u, \text{`red'})$, $R_2(v, \text{`blue'})$, and $R_3(u, \text{`red'}, v, \text{`blue'}, 1)$ of $\mathcal{J}$ (for some vertex $u \neq v \in V$) do not appear in $Q(\mathcal{J'})$. Notice that each of these summands contributes value 1 to $Q(\mathcal{J})$ and the number of these summands is at most  $|E|$. At the same time, the summand that contributes value $m$ to $Q(\mathcal{J})$ arising from the tuples $R_1(v, \text{`red'})$, $R_2(v, \text{`blue'})$, and $R_3(v, \text{`red'}, v, \text{`blue'})$ of $\mathcal{J}$ also does not appear in $Q(\mathcal{J'})$. Since $m = -|E| - 1$, it follows that $Q(\mathcal{J'})$ cannot be made smaller than $Q(\mathcal{J})$ on account of such a node $v$.

\textbf{Case 2}: Let $v$ be a node such that $R_1(v, \text{`blue'}), R_2(v, \text{`red'}) \in \mathcal{J}$. In this case, while populating $\mathcal{J'}$, vertex $v$ changes its color in relation $R_1$, i.e., we have that $R_1(v, \text{`red'}) \in \mathcal{J'}$ and $R_1(v, \text{`blue'}) \not\in \mathcal{J'}$. Compared to $Q(\mathcal{J})$, this can only increase the number of summands that contribute 1 to $Q(\mathcal{J'})$, by possibly having new summands arising from the tuples of type $R_1(v, \text{`red'})$, $R_2(w, \text{`blue'})$, and $R_3(v, \text{`red'}, w, \text{`blue'}, 1)$ of $\mathcal{J'}$ (for some vertex $w \neq v \in V$). Moreover, for every vertex $u \in V$, it is true that, if $R_1(u, \text{`red'}) \in \mathcal{J}$ then $R_1(u, \text{`red'}) \in \mathcal{J'}$; similarly, if $R_2(u, \text{`blue'}) \in \mathcal{J}$ then $R_2(u, \text{`blue'}) \in \mathcal{J'}$. Therefore, every summand that contributes 1 to $Q(\mathcal{J})$ also contributes 1 to $Q(\mathcal{J'})$. Hence, $Q(\mathcal{J'})$ cannot be made smaller than $Q(\mathcal{J})$ on account of such a node $v$.

The preceding analysis implies that
$Q(\mathcal{J'}) \geq Q(\mathcal{J})$.
\end{proof}

By Lemma \ref{lemma1}, there exists a repair $\mathcal{J}$ of $\mathcal{I}$ such that  $\mathcal{J}$ produces a red-blue coloring of $G$ and $Q(\mathcal{J})$ is the $lub$-answer in $\cons{Q, \mathcal{I}}$. We will show that, for a non-negative integer $k$, there is a cut $(S, \overline{S})$ of $G$ such that $|E(S, \overline{S})| \geq  k$ if and only if there exists a repair $\mathcal{J}$ of $\mathcal{I}$ such that $\mathcal{J}$ produces a red-blue coloring of $G$ and $Q(\mathcal{J}) \geq  k$. Once this is shown, it will follow that it is NP-hard to even compute the $lub$-answer in $\cons{Q, \mathcal{I}}$. 


Let $(S,\overline{S})$ be a cut of $G$ such that
$|E(S,\overline{S})|  \geq k$. Construct an $\mathcal{R}$-instance $\mathcal{J}$ as follows. For each vertex $v \in S$, add tuples $R_1(v, \text{`red'})$ and $R_2(v, \text{`red'})$ to $\mathcal{J}$. For each vertex $v \in \overline{S}$, add tuples $R_1(v, \text{`blue'})$ and $R_2(v, \text{`blue'})$ to $\mathcal{J}$. Add all tuples from relation $R_3$ of $\mathcal{I}$ to $\mathcal{J}$. Observe that $\mathcal{J}$ is a repair of $\mathcal{I}$ and that $\mathcal{J}$ produces a red-blue coloring of $G$. Also, every edge $(u, v) \in E$ such that $u \in S$ and $v \in \overline{S}$ is part of a witness to  a summand that contributes 1 to $Q(\mathcal{J})$. Moreover, no summand in $Q(\mathcal{J})$ arises from  a tuple of the form $R_3(v, \text{`red'}, v, \text{`blue'}, m)$ for some $v \in V$. Since we have that $|E(S, \overline{S})| \geq  k$, it must be the case that $Q(\mathcal{J}) \geq k$.
In the other direction, let $\mathcal{J}$ be a repair of $\mathcal{I}$ such that $\mathcal{J}$  produces a red-blue coloring of $G$ and   $Q(\mathcal{J}) \geq  k$. 
Construct two sets $S$ and $\overline{S}$ of vertices of $G$ as follows. Let $v \in S$ if $R_1(v, \text{`red'}) \in \mathcal{J}$, and let $v \in \overline{S}$ if $R_1(v, \text{`blue'}) \in \mathcal{J}$. Clearly, $(S, \overline{S})$ is a cut of $G$. Every edge $(u, v) \in E$ such that $u \in S$ and $v \in \overline{S}$ is part of a witness to a summand that contributes 1 to $Q(\mathcal{J})$ since the tuples $R_1(u, \text{`red'})$, $R_2(v, \text{`blue'})$, and $R_3(u, \text{`red'}, v, \text{`blue'}, 1)$ of $\mathcal{J}$ satisfy the underlying conjunctive query of $Q$. In fact, since $\mathcal J$ produces a red-blue coloring of $G$, every summand that contributes to $Q(\mathcal{J})$ must arise from such tuples. Since $Q(\mathcal{J}) \geq  k$, it must be the case that $|E(S, \overline{S})| \geq  k$.
\end{proof}

\section{Consistent Answers via SAT Solving}\label{main-body}
In this section, we give  polynomial-time reductions from computing the range  consistent answers of aggregation queries to variants of \sat{}. The reductions in this section assume that the database schema has one key constraint per relation; in Section \ref{beyond-key-constraints}, we show how these reductions can  be extended to schemata with arbitrary denial constraints. Our reductions rely on several well-known \emph{optimization} variants of \sat{} that we describe next. 

\smallskip
$\bullet$~
    \weightedmaxsat{} (or \wmaxsat{}) is the maximization variant of \sat{} in which each clause is assigned a positive weight and the goal is to find  an assignment that maximizes the sum of the weights of the satisfied clauses. We will write $(l_1 \lor \cdots \lor l_k, w)$ to denote
     a clause $(l_1 \lor \cdots \lor l_k)$ with weight $w$.

\smallskip

$\bullet$~ \partialmaxsat{} (or \pmaxsat{}) is the maximization variant of \sat{} in which some clauses of the formula are assigned infinite weight (\textit{hard clauses}), while each of the rest is assigned weight one (\textit{soft clauses}). The goal is to find an assignment that   satisfies all hard clauses and the maximum number of soft clauses. If the hard clauses of a \pmaxsat{} instance are not simultaneously satisfiable, then we say that the instance is \emph{unsatisfiable}. 
For simplicity, a hard clause $(l_1 \lor \cdots \lor l_k, \infty)$ is  denoted as $(l_1 \lor \cdots \lor l_k)$.  

\smallskip

$\bullet$~   \weightedpartialmaxsat{} (or \wpmaxsat{}) is the maximization variant of \sat{} in which some of the clauses of the formula  
are assigned infinite weight (\textit{hard clauses}), while each of the rest is assigned a positive weight (\textit{soft clauses}).
The goal is to find an 
 assignment that satisfies all hard clauses and maximizes the sum of weights of the satisfied soft clauses. Clearly, \wpmaxsat{} is a common generalization of both 
 \wmaxsat{} (no hard clauses) and 
 \pmaxsat{} (each soft clause has weight one).

 Modern solvers, such as MaxHS \cite{Davies2011}, can efficiently solve large instances of 
 these maximization variants of \sat{}.
Note that these maximization problems have  dual minimization problems, called 
 \wminsat{}, \pminsat, and \wpminsat{}, respectively.  For example, in \wpminsat{}, the goal is to  find  an assignment that satisfies all hard clauses and minimizes the sum of weights of the satisfied soft clauses. These minimization problems are of interest to us, because some of the computations of the range consistent answers have natural reductions to such minimization problems.
 At present, the only existing \wpminsat{} solver is MinSatz \cite{Li11}. Since this solver has certain size limitations, we will deploy K\"{u}gel's technique \cite{Kuegel12} to first reduce \wpminsat{}  to \wpmaxsat{}, and then use the MaxHS solver in our experiments.
  This technique uses the concept of CNF-\emph{negation} \cite{Kuegel12, Zhu12}.
  By definition, if $C=(l_1 \lor \cdots \lor l_k)$ is a clause, then the CNF-negation $\text{CNF}(\overline{C})$ of $C$ is the CNF-formula  $\neg l_1 \land (l_1 \lor \neg l_2) \land \cdots \land (l_1 \lor l_2 \lor \cdots \lor l_{k-1} \lor \neg l_k)$.
  It is easy to verify that the following properties hold: (i) if an assignment $s$ does not satisfy $C$, then $s$ satisfies every clause of $\text{CNF}(\overline{C})$; (ii) if an assignment $s$ satisfies $C$, then $s$ satisfies all but one of the clauses of $\text{CNF}(\overline{C})$, namely, the clause
  $(l_1\lor \cdots \lor l_{j-1} \lor \neg l_j)$, where $j$ is the smallest index such that $s(l_j) = 1$. It follows that  $C$ is satisfiable if and only if $\text{CNF}(\overline{C})$ is unsatisfiable.
  For a weighted clause $C = (l_1 \lor \cdots \lor l_k, w)$, the CNF-negation $\text{CNF}(\overline{C})$ of $C$ is the formula  $(\neg l_1, w) \land (l_1 \lor \neg l_2, w) \land \cdots \land (l_1 \lor l_2 \lor \cdots \lor l_{k-1} \lor \neg l_k, w)$.

We also need to recall the notions of \emph{key-equal groups}
and \emph{bags of witnesses}. 
Let $\mathcal I$ be a database instance.

\smallskip

$\bullet$~
 We say that two facts of a relation $R$ of $\mathcal I$ are \textit{key-equal}, if they agree on the key attributes of $R$. A set $S$ of facts of  $\mathcal I$ is called a \textit{key-equal group} of facts if every two facts in $S$ are key-equal, and no fact in $S$ is key-equal to some fact in  ${\mathcal I}\backslash S$.
 
 \smallskip
 
$\bullet$~ Let $q({\bf z}):= \exists {\bf w}\; (R_1({\bf x}_1) \land ... \land R_m({\bf x}_m))$ be a conjunctive query, where each ${\bf x}_i$ is a tuple of variables and constants, and let ${\bf a} \in q(\mathcal I)$ be an answer to $q$ on $\mathcal I$.
Let \textit{vars}($q$) and \textit{cons}($q$) be the sets of variables and constants occurring in $q$.
A function $f : vars(q) \cup cons(q) \rightarrow vals({\mathcal I})$ is a \textit{witnessing assignment to $\bf{a}$} if the following hold: $f({\bf z}) = {\bf a}$;  if $x_j$ is a constant in $q$, then $f(x_j) = x_j$; and  if $R_i(x_1,\cdots,x_n)$ is an atom of $q$, then $R_i(f(x_1),\cdots,f(x_n))$ is a fact of $\mathcal I$.
We say that a  set $S$ of facts from $\mathcal I$ is a \textit{witness to {\bf a}} if there is a witnessing assignment $f$ to $\bf a$ such that $S = \{R_i(f(x_1), \cdots,f(x_n)):  R_i(x_1,\cdots,x_n)$ is an atom of $q\}$.
\ignore{
\begin{enumerate}[noitemsep,leftmargin=*]
    \item[i.] if $x_j$ is a constant in $q$, then $f(x_j) = x_j$,
    \item[ii.] $S = \{R_i(f(x_1), \cdots,f(x_n)):  R_i(x_1,\cdots,x_n)$ is an atom of $q\}$.
\end{enumerate}

Let $S$ be a set of facts of $\mathcal I$ and \textit{vals}$(S)$ the set of constants appearing as attribute values in the facts of $S$. }

 Note that two distinct witnessing assignments to an answer may give rise to the same witness. Thus, we consider the \textit{bag of witnesses to an answer}, i.e., the bag consisting of witnesses arising from all witnessing assignments to that answer, where each witness $S$ is accompanied by its multiplicity, an integer denoting the number of witnessing assignments that gave rise to $S$. Finally, we define the \textit{bag of witnesses to a conjunctive query} as the bag union of the bags of witnesses over all answers to $q$ on $\mathcal I$ (in the bag union the multiplicities of the same set are added). The \textit{bag of witnesses to a union} $q:= q_1\cup \cdots \cup q_k$ of conjunctive queries is the bag union of the bags of witnesses to each conjunctive query $q_i$ in $q$. The bag of witnesses will be used in computing the range consistent answers to aggregation queries. In effect, the bag of witnesses corresponds to the provenance polynomials of conjunctive queries and their unions \cite{DBLP:conf/pods/GreenKT07,DBLP:journals/sigmod/KarvounarakisG12}.

It is easy to verify that both the key equal groups and the bag of the witnesses can be computed using SQL queries. In Section \ref{sec:sum-count} and Section \ref{sec:min-max}, we give reductions to compute the range consistent answers to aggregation queries without grouping, and in Section \ref{sec:with-grouping}, we describe an iterative algorithm that uses these reductions to handle aggregation queries with grouping.

\subsection{Answering Queries with \texttt{SUM} and \texttt{COUNT}}\label{sec:sum-count}
Let $\mathcal{R}$ be a database schema with one key constraint per relation, and $Q$ be the aggregation query\smallskip\\
\centerline{$Q :=\; \texttt{SELECT } f \texttt{ FROM } T(U, A),$}\smallskip\\
where $f$ is one of the operators \texttt{COUNT(*)}, \texttt{COUNT($A$)}, \texttt{SUM($A$)}, 
and $T(U, A)$ is a relation expressed as a union of conjunctive queries over $\mathcal{R}$. We now reduce the range consistent answers $\cons{Q}$  of $Q$ to \pmaxsat{} and to \wpmaxsat{}.

\subsubsection{Reductions to \pmaxsat{} and \wpmaxsat{}}
\begin{reduction}\label{reduction1}
Let $Q := \texttt{SELECT $f$ FROM } T(U, A)$ be an aggregation query, where $f$ is one of the operators \texttt{COUNT(*)}, \texttt{COUNT($A$)}, and \texttt{SUM($A$)}. Let $\mathcal{I}$ be an $\mathcal R$-instance and $\mathcal{G}$ be the set of key-equal groups of facts of $\mathcal{I}$. For each fact $f_i$ of $\mathcal{I}$, introduce a boolean variable $x_i$. Let $\mathcal{W}$ be the bag of witnesses to  the query $q^*$ on
$\mathcal I$, where
\begin{equation*}
 q^* :=
    \begin{cases}
      \exists U \exists A\; T(U, A) & \text{if } f \text{ is } \texttt{COUNT(*)}\\
    \exists U \; T(U, A)  & \text{if } f \text{ is } \texttt{COUNT($A$)} \text{ or } \texttt{SUM($A$)}.
 \end{cases}
\end{equation*}

Construct a  partial CNF-formula $\phi$ (if $f$ is \texttt{COUNT(*)} or  \texttt{COUNT($A$)}) or a weighted partial CNF-formula $\phi$ (if $f$ is \texttt{SUM(A)})  as follows:
\begin{enumerate}
\item [(1)] For each $G_j \in \mathcal{G}$,
\begin{itemize}
\item construct a hard clause $\alpha_j := \underset{f_i \in G_j}{\lor} x_i$.
\item for each pair $(f_m, f_n)$ of facts in $G_j$ such that $m \neq n$, construct a hard clause $\alpha_j^{mn} := (\neg x_m \lor \neg x_n)$.
\end{itemize}
\item [(2a)] 
If $f$ is 
\texttt{COUNT(*)} or \texttt{COUNT($A$)}, then
for each witness $W_j \in \mathcal{W}$, 
construct a soft  clause $\beta_j$,  where $$\beta_j = \bigg(\underset{f_i \in W_j}{\lor} \neg x_i, m_j\bigg).$$
Construct a  partial CNF-instance $$\phi = \bigg(\overset{|\mathcal{G}|}{\underset{j=1}{\land}}\alpha_j\bigg)\land\bigg(\overset{|\mathcal{G}|}{\underset{j=1}{\land}}\bigg(\bigwedge_{\substack{f_m \in \mathcal{G}_j \\f_n \in \mathcal{G}_j}}\alpha^{mn}_j\bigg)\bigg)\land\bigg(\overset{|\mathcal{W}|}{\underset{j=1}{\land}}\beta_j\bigg).$$

\item [(2b)]
If $f$ is \texttt{SUM($A$)},
let $\mathcal{W}_P$ and $\mathcal{W}_N$ be the subsets of $\mathcal{W}$ such that for each $W_j \in \mathcal{W}$, we have $W_j \in \mathcal{W}_P$ iff $q^*(W_j) > 0$, and $W_j \in \mathcal{W}_N$ iff $q^*(W_j) < 0$.
Let also $w_j = m_j *||q^*(W_j)||$,  where $||q^*(W_j)||$ is the absolute value of $q^*(W_j)$. Construct a weighted soft clause $\beta_j$ and a conjunction $\gamma_j$ of hard clauses as follows. If  $W_j \in \mathcal{W_N}$, introduce a new variable $y_j$ and let
\begin{align*}
    \beta_j &= (y_j, w_j) \text{ and }\\ 
    \gamma_j&=\bigg(\Big(\underset{f_i \in W_j}{\lor} \neg x_i\Big)\lor y_j\bigg) \land \bigg(\underset{f_i \in W_j}{\land}(\neg y_j \lor x_i)\bigg);
\end{align*}
otherwise, let $\beta_j = \bigg(\underset{f_i \in W_j}{\lor} \neg x_i, w_j\bigg)$ and do not construct $\gamma_j$.\\
Construct a weighted partial CNF-instance $$\phi = \bigg(\overset{|\mathcal{G}|}{\underset{j=1}{\land}}\alpha_j\bigg)\land\bigg(\overset{|\mathcal{G}|}{\underset{j=1}{\land}}\bigg(\bigwedge_{\substack{f_m \in \mathcal{G}_j \\f_n \in \mathcal{G}_j}}\alpha^{mn}_j\bigg)\bigg)\land\bigg(\overset{|\mathcal{W}|}{\underset{j=1}{\land}}\beta_j\bigg)\land\bigg({\underset{W_j \in \mathcal{W_N}}{\land}}\gamma_j\bigg).$$

\ignore{
\item For each $W_j \in \mathcal{W}$, let $w_j=1$ if $f$ is 
\texttt{COUNT(*)} or \texttt{COUNT($A$)},
and let $w_j = ||q^*(W_j)||$ if $f$ is \texttt{SUM($A$)}, where $||q^*(W_j)||$ is the absolute value of $q^*(W_j)$. Construct a weighted soft clause $\beta_j$ (and a conjunction $\gamma_j$ of hard clauses) as follows. If $f$ is \texttt{SUM($A$)} and $W_j \in \mathcal{W_N}$, introduce a new variable $y_j$ and let
\begin{align*}
    \beta_j &= (y_j, w_j) \text{ and }\\ 
    \gamma_j&=\bigg(\Big(\underset{f_i \in W_j}{\lor} \neg x_i\Big)\lor y_j\bigg) \land \bigg(\underset{f_i \in W_j}{\land}(\neg y_j \lor x_i)\bigg);
\end{align*}
otherwise, let $\beta_j = \bigg(\underset{f_i \in W_j}{\lor} \neg x_i, w_j\bigg)$.
}

\end{enumerate}

\end{reduction}

\noindent\textbf{Purpose of the components of  $\phi$ in Reduction \ref{reduction1}}
\begin{itemize}[noitemsep,leftmargin=*]
    \item Each $\alpha_j$-clause encodes the \textit{``at-least-one"} constraint for each key-equal group $G_j$ in the sense that satisfying $\alpha_j$ requires setting at least one variable corresponding to a fact in $G_j$ to true. Similarly, each $\alpha_j^mn$-clause encodes the \textit{``at-most-one"} constraint for $G_j$. In effect, every assignment that satisfies all $\alpha$-clauses sets exactly one variable corresponding to the facts from each key-equal group to true, and thus uniquely corresponds to a repair of $\mathcal I$.
    \item Satisfying a $\beta_j$-clause constructed in Step 2a requires setting at least one variable corresponding to the facts of a witness $W_j$ to $q^*$ on $\mathcal I$ to false. Thus, if $s$ is an assignment that satisfies all $\alpha$-clauses, then $\beta_j$ is satisfied by $s$ if and only if $W_j \not\in \mathcal J$, where $\mathcal J$ is a repair corresponding to $s$.
    \item The $\beta_j$-clauses constructed in Step 2b serve the same purpose as the ones from Step 2a, but here they are constructed only for the witnesses in $\mathcal W_P$. For the witnesses in $\mathcal W_N$, the $\beta_j$-clauses encode the condition that $\beta_j$ is satisfied if and only if all variables corresponding to the facts in $W_j$ are set to true. The hard $\gamma_j$-clauses are used solely to express the equivalence $y_j \leftrightarrow (\underset{f_i \in W_j}{\land} x_i)$ in conjunctive normal form. 
\end{itemize}

The number of $\alpha$-clauses is $O(n)$, where $n$ is the size of the database; the number of $\beta$-clauses and $\gamma$-clauses combined is $O(n^k)$, where $k$ is the number of relation symbols in $Q$.

\begin{proposition}\label{prop1}
Let $Q := \texttt{SELECT $f$ FROM } T(U, A)$ be an aggregation query, where $f$ is one of the operators \texttt{COUNT(*)}, \texttt{COUNT($A$)}, and \texttt{SUM($A$)}. In a maximum (a minimum) satisfying assignment of the \wpmaxsat{}-instance $\phi$ constructed using Reduction \ref{reduction1}, the sum of weights of the falsified clauses is the glb-answer (lub-answer) in the range consistent answers $\cons{Q}$ on $\mathcal{I}$.
\end{proposition}
\begin{proof}
Let $Q := \texttt{SELECT COUNT(*) FROM } T(U, A)$, and let $s$ be an assignment of the formula $\phi$ constructed using Reduction \ref{reduction1}. Let $g(\phi, s)$ denote the sum of weights of the soft clauses of $\phi$ satisfied by $s$. Construct a database sub-instance $\mathcal{J}$ from $s$ such that $f_i \in \mathcal{J}$ if and only if $s(x_i) = 1$. The hard clauses of $\phi$ constructed in Step (1) of Reduction \ref{reduction1} encode the condition that exactly one fact from each key-equal group of facts of $\mathcal{I}$ is in $\mathcal{J}$, ensuring that $\mathcal{J}$ is a repair of $\mathcal{I}$. Moreover, the soft clauses of $\phi$ falsified by $s$ have a one-to-one correspondence with the witnesses to $q^*$ in $\mathcal{J}$. Therefore, we have that $g(\phi, s) = |\mathcal{W}| - Q(\mathcal{J})$. Since $|\mathcal{W}|$ does not depend on $\mathcal{J}$, the answer $Q(\mathcal{J})$ is minimized (i.e., $Q(\mathcal{J})$ is a $glb$-answer in $\cons{Q}$ on $\mathcal{I}$) when $s$ is a maximum satisfying assignment. Essentially the same argument works for the case where  $Q := \texttt{SELECT COUNT($A$) FROM } T(U, A)$.
A dual argument to this proves that a repair of $\mathcal{I}$ constructed from a minimum satisfying assignment of $\phi$ realizes the $lub$-answer in $\cons{Q}$ on $\mathcal{I}$.

Now, let $Q := \texttt{SELECT SUM($A$) FROM } T(U, A)$. Construct a repair $\mathcal{J}$ of $\mathcal{I}$ from $s$ by choosing $f_i \in \mathcal{J}$ if and only if $s(x_i) = 1$. Also, construct a database instance $\mathcal{I}_p$ as follows. For every fact $f \in \mathcal{I}$, let $f \in \mathcal{I}_p$ if and only if $f \in W_j$ for some $W_j \in \mathcal{W}_p$. Thus, $Q(\mathcal{I}_p)$ is the sum of values of the aggregation attribute evaluated on the witnesses in $\mathcal{W}_p$. Observe that, for every $W_j \in \mathcal{W}_P$, the clause $\beta_j$ is falsified by $s$ if and only if $W_j \in \mathcal{J}$. Similarly, for every $W_j \in \mathcal{W}_N$, the clause $\beta_j$ is satisfied by $s$ if and only if $W_j \in \mathcal{J}$. Therefore, we have that, 
\begin{align*}
Q(\mathcal{J}) &= Q(\mathcal{I}_p) - \Sigma_{W_j \in \mathcal{W}_P \land s(\beta_j) = 1} (w_j) - \Sigma_{W_j \in \mathcal{W}_N \land s(\beta_j) = 1} (w_j)\\
&=Q(\mathcal{I}_p) - g(\phi,s)
\end{align*}
Since $Q(\mathcal{I}_p)$ does not depend on $\mathcal{J}$, the answer $Q(\mathcal{J})$ is minimized (i.e., $Q(\mathcal{J})$ is a $glb$-answer in $\cons{Q}$ on $\mathcal{I}$) when $s$ is a maximum satisfying assignment. A dual argument to this proves that a repair of $\mathcal{I}$ constructed from a minimum satisfying assignment of $\phi$ realizes the $lub$-answer in $\cons{Q}$ on $\mathcal{I}$.
\end{proof}

\ignore{
\begin{corollary}\label{cor1}
Let $Q := \texttt{SELECT $f$ FROM } T(U, w)$, where $f$ is either \texttt{COUNT(*)} or \texttt{COUNT($w$)}. In a minimum satisfying assignment to the \pmaxsat{} instance $\phi$ constructed using Reduction \ref{reduction1}, the number of falsified clauses is the LUB-answer in $\cons{Q}$ on $\mathcal{I}$.
\end{corollary}
\begin{proof}
It follows from a dual argument to that of Proposition \ref{prop1} that a repair of $\mathcal{I}$ constructed from a minimum satisfying assignment to $\phi$ realizes the $LUB$-answer in $\cons{Q}$ on $\mathcal{I}$.
\end{proof}
}

\begin{example}\label{exmp1}
Let $\mathcal{I}$ be a database instance from Table \ref{toydb}, and $Q$ be the following aggregation query which counts the number of customers who have an account in their own city:\ignore{\smallskip\\
\texttt{\tab[1cm]SELECT COUNT(*)\\\tab[1cm]FROM CUST, ACC, CUSTACC\\\tab[1cm]WHERE CUST.CID = CUSTACC.CID\\\tab[2cm]AND ACC.ACCID = CUSTACC.ACCID\\\tab[2cm]AND CUST.CITY = ACC.CITY}\smallskip\\}
{\smallskip\\
\texttt{\tab[1cm]SELECT COUNT(*)\\\tab[1cm]FROM CUST, ACC, CUSTACC\\\tab[1cm]WHERE CUST.CID = CUSTACC.CID\\\tab[1.5cm]AND ACC.ACCID = CUSTACC.ACCID\\\tab[1.5cm]AND CUST.CITY = ACC.CITY}\smallskip\\}
From Reduction \ref{reduction1}, we construct the following clauses:
\begin{itemize}[leftmargin=0cm]
\item[] $\alpha$-clauses: $x_1, (x_2 \lor x_3), x_4, x_5, x_6, x_7, (x_8 \lor x_9), x_{10}$;
\item[] $\alpha^{mn}$-clauses: $(\neg x_2 \lor \neg x_3), (\neg x_8 \lor \neg x_9)$;
\item[] $\beta$-clauses: $(\neg x_1 \lor \neg x_6, 1), (\neg x_2 \lor \neg x_7, 1), (\neg x_3 \lor \neg x_9, 1)$.
\end{itemize}
Observe that it is okay to omit the variables corresponding to the facts in \texttt{CUSTACC} since \texttt{CUSTACC} does not violate $\Sigma$. A maximum satisfying assignment to the \pmaxsat{} instance $\phi$ constructed from above clauses is $x_i =0$ for $i \in \{2, 9\}$, and $x_i = 1$ otherwise. It falsifies one clause, namely, $(\neg x_1 \lor \neg x_6, 1)$. Similarly, an assignment $x_i =0$ for $i \in \{2, 8\}$, and $x_i = 1$ otherwise is a minimum satisfying assignment to the \pminsat{} instance $\phi$, and it falsifies two clauses, namely, $(\neg x_1 \lor \neg x_6, 1)$ and $(\neg x_3 \lor \neg x_9, 1)$. Thus, $\cons{Q, \mathcal{I}}$ w.r.t.\ range semantics is $[1, 2]$ by Proposition \ref{prop1}.
\end{example}

\begin{example}\label{exmp2}
Let us again consider the database instance $\mathcal{I}$ from Table \ref{toydb}, and the following aggregation query $Q$:\ignore{\smallskip\\
\texttt{\tab[1cm]SELECT SUM(ACC.BAL)\\\tab[1cm]FROM CUST, ACC, CUSTACC\\\tab[1cm]WHERE CUST.CID = CUSTACC.CID\\\tab[2cm]AND ACC.ACCID = CUSTACC.ACCID\\\tab[2cm]AND CUST.CNAME = `Mary'}\smallskip\\}
{\smallskip\\
\texttt{\tab[1cm]SELECT SUM(ACC.BAL)\\\tab[1cm]FROM CUST, ACC, CUSTACC\\\tab[1cm]WHERE CUST.CID = CUSTACC.CID\\\tab[1.5cm]AND ACC.ACCID = CUSTACC.ACCID\\\tab[1.5cm]AND CUST.CNAME = `Mary'}\smallskip\\}
The hard clauses constructed using Reduction \ref{reduction1} are same as the ones from Example \ref{exmp1}. The rest of the clauses are as follows:
\begin{itemize}[leftmargin=0cm]
\item[] $\beta$-clauses: $(\neg x_2 \lor \neg x_7, 1000)$, $(\neg x_3 \lor \neg x_7, 1000)$, $(\neg x_2 \lor \neg x_8, 1200)$, $(\neg x_3 \lor \neg x_8, 1200)$, $(y_1, 100)$, $(y_2, 100)$.
\item[] $\gamma$-clauses: $(\neg x_2 \lor \neg x_9 \lor y_1)$, $(\neg y_1 \lor x_2)$, $(\neg y_1 \lor x_9)$, $(\neg x_3 \lor \neg x_9 \lor y_2)$, $(\neg y_2 \lor x_3)$, $(\neg y_2 \lor x_9)$.
\end{itemize}
The witnesses $\{f_2, f_9, f_{13}\}$ and $\{f_3, f_9, f_{13}\}$ belong to $\mathcal{W}_N$ because the account balance is -100 in both cases, so we introduce new variables $y_1$ and $y_2$ respectively, and construct hard $\gamma$-clauses as described above. The $\beta$-clauses corresponding to these witnesses are $(y_1, 100)$ and $(y_2, 100)$. We omit $x_{13}$ in all of these clauses since \texttt{CUSTACC} does not violate $\Sigma$. Note that $Q(\mathcal{I}_p) = 4400$. An assignment in which $x_8 = 0$ and $x_9 =1$ is a maximum satisfying assignment to the \pmaxsat{} instance $\phi$ constructed. The sum of satisfied soft clauses by this assignment is 3500 since it satisfies two clauses with weights 1200 each, one with weight 1000 and one with weight 100. Thus, by Proposition \ref{prop1}, we have that $Q(\mathcal{J}) = 4400 - 3500 = 900$ where $\mathcal{J}$ is a repair corresponding to the assignment in consideration. Similarly, setting $x_8 = 1$ and $x_9 = 0$ yields a minimum satisfying assignment in which the sum of satisfied soft clauses is 2200, since it satisfies one clause with weight 1200 and one with weight 1000, indicating that $\cons{Q, \mathcal{I}} = [900, 2200]$.
\end{example}

\subsubsection{Handling \texttt{DISTINCT}}

Let $Q := \texttt{SELECT $f$ FROM } T(U, A)$ be an aggregation query, where $f$ is either \texttt{COUNT(DISTINCT $A$)} or \texttt{SUM(DISTINCT $A$)}. Solving a \pmaxsat{} or a \wpmaxsat{} instance constructed using Reduction \ref{reduction1} may yield incorrect $glb$ and $lub$ answers to $Q$, if the database contains multiple witnesses with the same value for attribute $A$. For example, consider the database instance $\mathcal{I}$ from Table \ref{toydb}, and a query\smallskip\\
\centerline{$Q:= \texttt{SELECT COUNT(DISTINCT ACC.TYPE) FROM ACC}$}\smallskip\\
The correct $glb$ and $lub$-answers in $\cons{Q, \mathcal{I}}$ are both 2, but solutions to the \pmaxsat{} and \pminsat{} instances constructed using Reduction \ref{reduction1} yield both answers as 4. The reason behind this is that the soft clauses $\neg x_6$ and $\neg x_7$ both correspond to the account type Checking, and similarly $\neg x_8$, $\neg x_9$, and $\neg x_{10}$ all correspond to the account type Saving. The hard clauses in the formula ensure that $x_6$, $x_7$, $x_{10}$, and one of $x_8$ and $x_9$ are true, thus counting both Checking and Saving account types exactly twice in every satisfying assignment to the formula. This can be handled by modifying the $\beta$-clauses in Reduction \ref{reduction1} as follows.

Let $\mathcal{A}$ denote a set of distinct answers to the query $q^*(A):= \exists U \; T(U, A)$. For each answer $b \in \mathcal{A}$, let $\mathcal{W}^b$ denote a subset of $\mathcal{W}$ such that for every witness $W \in \mathcal{W}^b$, we have that $q^*(W) = b$. The idea is to use auxiliary variables to construct one soft clause for every distinct answer $b \in \mathcal{A}$, such that it is true if and only if no witness in $\mathcal{W}^b$ is present in a repair corresponding to the satisfying assignment. First, for every witness $W^b_j \in \mathcal{W}^b$, we introduce an auxiliary variable $z^b_j$ that is true if and only if $W^b_j$ is not present in the repair. Then, we introduce an auxiliary variable $v^b$ which is true if and only if all $z^b$-variables are true. These constraints are encoded in the set $H^b$ returned by Algorithm \ref{alg:handle-distinct}, and are forced by making clauses in $H^b$ hard. For every answer $b \in \mathcal{A}$, Algorithm \ref{alg:handle-distinct} also returns one $\beta^b$-clause, which serves the same purpose as the $\beta$-clauses in Reduction \ref{reduction1}. Now, a \pmaxsat{} or a \wpmaxsat{} instance can be constructed by taking in conjunction all $\alpha$-clauses from the key-equal groups, the hard $\gamma$-clauses if any, the hard clauses from all $H^b$-sets, and all soft $\beta^b$-clauses. With this, it is easy to see that a maximum (\textit{or minimum}) satisfying assignment to \pmaxsat{} or \wpmaxsat{} instance give us the $glb$-answer (\textit{or $lub$-answer}) in $\cons{Q}$. This is illustrated in Example \ref{exmp3}.

\begin{algorithm}
\caption{Handling \texttt{DISTINCT}}\label{alg:handle-distinct}
\begin{algorithmic}[1]
\Procedure{handleDistinct}{$\mathcal{W}^b$}
\State \textbf{let} $H^b = \emptyset$\hfill //Empty set of clauses
\For{$W^b_j \in \mathcal{W}^b$}
\State $H^b = H^b \bigcup \Big\{\Big(\neg z^b_j \lor \Big(\underset{f_i \in W^b_j}{\bigvee} \neg x_i\Big)\Big)\Big\}$
\For{$f_i \in W^b_j$}
\State $H^b = H^b \bigcup\; \{(z^b_j \lor x_i)\}$
\EndFor
\EndFor
\State $H^b = H^b \bigcup \Big\{\Big(\neg v^b \lor \Big(\underset{W^b_j \in \mathcal{W}^b}{\bigvee} \neg z^b_j\Big)\Big)\Big\}$
\For{$W^b_j \in \mathcal{W}^b$}
\State $H^b = H^b \bigcup\; \{(\neg v^b \lor z^b_j)\}$
\EndFor
\State \textbf{let} $\beta^b = (v^b, 1)$
\If{($f$ is \texttt{SUM(DISTINCT $A$)})}
\State $\beta^b = (v^b, ||b||)$
\If{$b < 0$} $\beta^b = (\neg v^b, ||b||)$
\EndIf
\EndIf
\State \Return $H^b, \beta^b$
\EndProcedure
\end{algorithmic}
\end{algorithm}

\begin{example}\label{exmp3}
Consider the following aggregation query $Q$ on the database instance $\mathcal{I}$ from Table \ref{toydb}:\smallskip\\
\texttt{\tab[1cm]SELECT COUNT(DISTINCT ACC.TYPE) FROM ACC}\smallskip\\
We have that $\mathcal{A} = \{\text{`Checking'}, \text{`Saving'}\}$. Let us denote these two answers by $a_1$ and $a_2$ respectively. Since every witness to the query consists of a single fact, every $y^a$-variable is equivalent to a single literal, for example, $y^{a_1}_1 \leftrightarrow \neg x_6$ and $y^{a_1}_2 \leftrightarrow \neg x_7$. As a result, it is unnecessary to introduce any $y^a$-variables at all. Thus, we construct the following clauses from Reduction \ref{reduction1} and Algorithm \ref{alg:handle-distinct}:
\begin{itemize}[leftmargin=0cm]
\item[] $\alpha$-clauses: $x_6, x_7, (x_8 \lor x_9), x_{10}$; $\alpha^{mn}$-clauses: $(\neg x_8 \lor \neg x_9)$;
\item[] $H^{a_1}: (x_6 \lor x_7 \lor v^{a_1}), (\neg v^{a_1} \lor \neg x_6), (\neg v^{a_1} \lor \neg x_7)$;
\item[] $H^{a_2}: (x_8 \lor x_9 \lor x_{10} \lor v^{a_2}), (\neg v^{a_2} \lor \neg x_8), (\neg v^{a_2} \lor \neg x_9), (\neg v^{a_2} \lor \neg x_{10})$;
\item[] $\beta$-clauses: $(v^{a_1}, 1), (v^{a_2}, 1)$
\end{itemize}
The maximum and minimum satisfying assignments to the \pmaxsat{} and \pminsat{} instances constructed using these clauses falsify both $\beta$-clauses, since $\cons{Q, \mathcal{I}}$ w.r.t.\ range semantics is $[2, 2]$.
\end{example}

\subsection{Answering Queries with \texttt{MIN} and \texttt{MAX}}\label{sec:min-max}
Let $\mathcal{R}$ be a database schema with one key constraint per relation, and $Q$ be the aggregation query\smallskip\\
\centerline{$Q :=\; \texttt{SELECT } f \texttt{ FROM } T(U, A),$}\smallskip\\ where $f$ is one of the operators \texttt{MIN($A$)} and \texttt{MAX($A$)}, and $T(U, A)$ is a relation expressed as a union of conjunctive queries over $\mathcal{R}$. The semantics of the range consistent answers to aggregation queries with \texttt{MIN} and \texttt{MAX} operators are similar to aggregation queries with \texttt{SUM} or \texttt{COUNT} operators, but here we need to address one additional special case. We illustrate this special case using Example \ref{exmp:special-case-min-max}.

\begin{example}\label{exmp:special-case-min-max}
Consider the database instance $\mathcal I$ from Table \ref{table:run} and two aggregation queries $Q_1$ and $Q_2$ as follows.\smallskip

\texttt{$Q_1:=$ SELECT SUM(ACCOUNTS.BAL) FROM ACCOUNTS\\\tab[1.25cm]WHERE ACCOUNTS.CITY = `SF'}\smallskip

\texttt{$Q_2:=$ SELECT MIN(ACCOUNTS.BAL) FROM ACCOUNTS\\\tab[1.25cm]WHERE ACCOUNTS.CITY = `SF'}

It is clear that the range consistent answers to $\cons{Q_1, \mathcal I} = [-100, 0]$. The \textit{lub}-answer of 0 in $\cons{Q_1, \mathcal I}$ comes from a repair on which there is no account in the city of SF, and therefore the \texttt{SUM} function returns 0. For $Q_2$, however, the range consistent answers are unclear because the \texttt{MIN} function is not defined for empty sets.
\end{example}

In such scenarios, various different semantics can be considered. One natural semantics is that if there exists a repair on which the underlying conjunctive query evaluates to an empty set, we could say that there is no consistent answer to the aggregation query. Another one could be to return the interval [\textit{glb}, \textit{lub}] of values that come from the repairs on which the underlying conjunctive query evaluates to a non-empty set of answers, and additionally return the information about the existence of the repair on which the underlying conjunctive query returns the empty set of answers. The reductions we give in this Section can be used regardless of which of the two above-mentioned semantics is chosen.

In what follows, we first show that the \textit{glb}-answer to an aggregation query with the \texttt{MIN($A$)} operator can be computed in polynomial time in the size of the original inconsistent database instance $\mathcal I$ (Proposition \ref{prop:trivial-bound-min}). We then give an iterative \sat{}-based approach to compute the \textit{lub}-answer to an aggregation query with the \texttt{MIN($A$)} operator. We do not explicitly state methods to obtain the range consistent answers aggregation queries with the \texttt{MAX($A$)} operator since it is straightforward that the \textit{lub}-answer for \texttt{MIN($A$)} is a dual of the \textit{glb}-answer for \texttt{MAX($A$)} in the sense that computing the \textit{lub}-answer for the \texttt{MIN($A$)} operator yields the same result as negating all values of the aggregation attribute $A$ in the database and then computing the \textit{glb}-answer for the \texttt{MAX($A$)} operator.

\begin{proposition}\label{prop:trivial-bound-min}
Let $\mathcal R$ be a database schema, $\mathcal I$ an $\mathcal R$-instance, and $Q$ the aggregation query $\texttt{SELECT MIN($A$) FROM } T(U, A)$. Let $q_1$ be the union $q_1(A):= \exists U\; T(U, A)$ of conjunctive queries and $W_{\textit{glb}}$ be the witness to $q_1$ on $\mathcal I$ such that no two facts in $W_{\textit{glb}}$ are key-equal, and there is no $W'$ such that $q(W') < q(W_{\textit{glb}})$ and no two facts in $W'$ are key-equal. Then, $q_1(W_{\textit{glb}})$ is the \textit{glb}-answer in $\cons{Q, \mathcal I}$.
\end{proposition}
\begin{proof}
For every witness $W'$ to $q_1$ on $\mathcal I$ such that $q_1(W') < q_1(W_{\textit{glb}})$, we have that no repair of $\mathcal I$ contains $W'$ because $W'$ contains at least two key-equal facts. Moreover, since no two facts in $W_{\textit{glb}}$ are key-equal, there exists a repair $\mathcal J$ of $\mathcal I$ such that $W_{\textit{glb}} \in \mathcal J$. Therefore, $q_1(W_{\textit{glb}})$ must be the smallest possible answer to $Q$ on $\mathcal I$, i.e., the \textit{glb}-answer in $\cons{Q, \mathcal I}$. Since the number of witnesses to $q_1$ is polynomial in the size of $\mathcal I$, a desired witness $W_{\textit{glb}}$ can be obtained efficiently from the result of evaluating $q_1$ on $\mathcal I$.
\end{proof}
\ignore{
\begin{corollary}\label{cor:trivial-bound-max}
Let $\mathcal R$ be a database schema, $\mathcal I$ an $\mathcal R$-instance, and $Q$ the aggregation query $\texttt{SELECT MAX($A$) FROM } T(U, A)$. Let $q_1$ be the union $q_1(A):= \exists U\; T(U, A)$ of conjunctive queries and $W_{\textit{lub}}$ be the witness to $q_1$ on $\mathcal I$ such that no two facts in $W_{\textit{lub}}$ are key-equal, and there is no $W'$ such that $q(W') > q(W_{\textit{lub}})$ and no two facts in $W'$ are key-equal. Then, $q_1(W_{\textit{lub}})$ is the \textit{lub}-answer in $\cons{Q, \mathcal I}$.
\end{corollary}
\begin{proof}
It follows from a dual argument to that of Proposition \ref{prop:trivial-bound-min}.
\end{proof}}

To compute the range consistent answers to aggregation queries with \texttt{MIN($A$)} and \texttt{MAX($A$)} operators, we opt for an iterative \sat{} solving approach. In what follows, we formalize the construction of the \sat{} instance for the first iteration (Construction \ref{const:min-iterative-sat}) and give Algorithm \ref{alg:min-iterative-sat} that computes the \textit{lub}-answer in $\cons{Q, \mathcal I}$ by constructing and solving \sat{} instances in subsequent iterations.

\begin{construction}\label{const:min-iterative-sat}
Given an $\mathcal R$-instance $\mathcal I$, construct a CNF formula $\phi$ as follows. For each fact $f_i$ of $\mathcal I$, introduce a boolean variable $x_i$. Let $\mathcal{G}$ be the set of key-equal groups of facts of $\mathcal I$, and $\mathcal W = \{W_1, \cdots, W_m\}$ denote the set of minimal witnesses to a conjunctive query $q$ on $\mathcal I$, where $q(w) := \exists \vec{u}\; T(\vec{u}, w)$. Assume that the set $\mathcal W$ is sorted in descending order of the answers, i.e., for $1 \leq i < m$, we have that $q(W_i) \geq q(W_{i+1})$.
\begin{itemize}[noitemsep]
\item For each $G_j \in \mathcal{G}$, construct a clause $\alpha_j = \underset{f_i \in G_j}{\lor} x_i$.
\item Construct a CNF formula $\phi = \overset{|\mathcal{G}|}{\underset{j=1}{\land}}\alpha_j$.
\end{itemize}
\end{construction}

\begin{algorithm}[ht]
\caption{Computing the \textit{lub}-answer in $\cons{Q, \mathcal I}$ for \texttt{MIN} via Iterative \sat{}}\label{alg:min-iterative-sat}
{\fontsize{11}{12}\selectfont
\begin{algorithmic}[1]
\Procedure{LubAnswer-IterativeSAT}{$\phi, \mathcal W$}
\State \textbf{let} $v = q(W_1), j = 1$
\While{$j \leq |\mathcal W|$}
\If{$v = q(W_j)$}
\State \textbf{let} $\phi = \phi \land \Big({\underset{f_i \in W_j}{\lor}}\neg x_i\Big)$
\State \textbf{let} $j = j + 1$
\Else
\If{\textsc{UNSAT}$(\phi)$}
\State \textbf{return} $q(W_{j-1})$
\Else \State \textbf{let} $v = q(W_j)$
\EndIf
\EndIf
\EndWhile
\State \textbf{return} $q(W_{|\mathcal W|})$
\EndProcedure
\end{algorithmic}}
\end{algorithm}
\begin{proposition}\label{prop4}
Let $Q := \texttt{SELECT MIN($A$) FROM } T(U, A)$ be an aggregation query, and $\mathcal{I}$ be a database instance. Algorithm \ref{alg:min-iterative-sat} returns the \textit{lub}-answer in $\cons{Q, \mathcal{I}}$.
\end{proposition}
\begin{proof}
The $\alpha$-clauses of $\phi$ make sure that a repair $\mathcal{J}$ of $\mathcal{I}$ can be constructed from every assignment $s$ of $\phi$ that satisfies the $\alpha$-clauses, by arbitrarily choosing exactly one fact $f_i$ from each key-equal group of $\mathcal{I}$ such that $s(x_i) = 1$.
Let $\mathcal{A} = \{A_1, \cdots, A_{lub}, \cdots, A_{|\mathcal{A}|}\}$ denote the set of distinct answers to a conjunctive query $q(w) := \exists \vec{u} \; T(\vec{u}, w)$ on $\mathcal{I}$, where $A_{lub}$ is the \textit{lub}-answer to $Q$ on $\mathcal{I}$. For a witness $W$ to $q$, a clause $({\underset{f_i \in W}{\lor}}\neg x_i)$ is satisfied by an assignment $s$ if and only if $W$ is not present in any repair constructed from $s$. At iteration $j$ of the while-loop, if the formula $\phi$ is checked for satisfiability (line 8 of Algorithm \ref{alg:min-iterative-sat}), the formula contains the $\alpha$-clauses corresponding to the key-equal groups of $\mathcal I$ in conjunction to all clauses corresponding to the minimal witnesses to $q$ on which the $q$ evaluates to an answer strictly smaller than $q(W_j)$. At this point, if the formula $\phi$ is satisfiable, then there exists a repair $\mathcal J$ of $\mathcal I$ such that $q(W_{j-1}) \notin q(\mathcal J)$, and also for all potential answers $A_i \leq q(W_{j-1})$, we have that $A_i \notin q(\mathcal J)$. On the other hand, if the formula is unsatisfiable, then there exists no repair $\mathcal J$ of $\mathcal I$ such that $q(W_{j-1}) \notin q(\mathcal J)$ and $A_i \notin q(\mathcal J)$ for all $A_i \leq q(W_{j-1})$. Since the clauses are added in the ascending order of the answers, we have that $\phi$ satisfiable at iteration $j$ if and only if $q(W_{j-1}) < A_{lub}$. Therefore, if $\phi$ becomes unsatisfiable for the first time at iteration $j$, it must be the the case that $q(W_{j-1})$ is the \textit{lub}-answer in $\cons{Q, \mathcal I}$.
\end{proof}

\subsubsection*{Why not a Binary Search?}~
In essence, Algorithm \ref{alg:min-iterative-sat} works like a linear search on a sorted array. It may sound appealing to perform the binary search instead of the linear search for obvious reasons. Clearly, the \sat{} solver will only have to solve $O(log_2\; |\mathcal A|)$ instances of \sat{} instead of $O(|\mathcal A|)$ instances of \sat{}. The problem with this approach is the following. Observe that, on an average, half of the \sat{} instances that the solver needs to solve in the binary search approach will be unsatisfiable. In the linear search approach, however, all but the last instance given to the solver are satisfiable. Typically, the proofs of unsatisfiability produced by the \sat{} solvers are significantly large compared to the proofs of satisfiability as the unsatisfiability of an instance needs to be proven with a refutation tree that can be exponential in size of the formula, while just one satisfying assignment is enough to prove the satisfiability of an instance. As a result, \sat{} solvers typically take considerably long amounts of time to solve unsatisfiable instances but they are very quick on most real-world satisfiable instances. Therefore, in practice, the linear search often works better than the binary search.

\subsection{Answering Queries with Grouping}\label{sec:with-grouping}
Let $Q$ be the aggregation query
\smallskip\\\centerline{$Q :=\; \texttt{SELECT }Z, f \texttt{ FROM } T(U, Z, w) \texttt{ GROUP BY }Z,$}\smallskip\\
where $f$ is one of \texttt{COUNT}$(*)$, \texttt{COUNT}$(A)$, \texttt{SUM}$(A)$, \texttt{MIN}$(A)$, or \texttt{MAX}$(A)$, and $T(U, A)$ is a relation expressed by a union of conjunctive queries on $\mathcal{R}$. We refer to the attributes in $Z$ as the grouping attributes. For aggregation queries with grouping, it does not seem feasible to reduce $\cons{Q}$ to a single \pmaxsat{} or a \wpmaxsat{} instance because for each group of consistent answers, the GLB-answer and the LUB-answer may realize in different repairs of the inconsistent database. To illustrate this, consider the database from Table \ref{toydb} and a query $Q:= \texttt{SELECT COUNT(*) FROM CUST GROUP BY CUST.CITY}$. Notice that, the GLB-answers (LA, 2) and (SF, 1) in $\cons{Q}$ come from two different repairs of relation \texttt{CUST}, namely, $\{f_1, f_3, f_4, f_5\}$ and $\{f_1, f_2, f_4, f_5\}$ respectively.
However, the reductions from the preceding section can be used to compute the bounds to each consistent group of answers independently. For a given aggregation query $Q$ with grouping, we first compute the consistent answers to an underlying conjunctive query $q(Z) := \exists U, A\; T(U, Z, A)$. Then, for each answer $b$ in $\cons{q}$, we compute the GLB and LUB-answers to the query $Q' :=\; \texttt{SELECT } f \texttt{ FROM } T(U, Z, A) \land (Z = b)$ via \pmaxsat{} or \wpmaxsat{} solving as shown in Algorithm \ref{alg:grouping}.

\begin{algorithm}[h]
\caption{Consistent Answers to Queries With Grouping}
\begin{flushleft}
Let $\mathcal{I}$ be an inconsistent database instance, and $Q$ be an aggregation query of the form $Q :=\; \texttt{SELECT }Z, f \texttt{ FROM } T(U, Z, A) \texttt{ GROUP BY }Z$.
\end{flushleft}

\noindent\makebox[\linewidth]{\rule{\columnwidth}{0.4pt}} 

\begin{algorithmic}[1]
\Procedure{ConsAggGrouping}{$Q$}
\State \textbf{let} $Ans = \emptyset$
\State \textbf{let} $q(Z) := \exists U, A\; T(U, Z, A)$
\State \textbf{let} $\mathcal{A}_c = \cons{q, \mathcal{I}}$

\For{$b \in \mathcal{A}_c$}
\State \textbf{let} $Q' :=\; \texttt{SELECT } f \texttt{ FROM } T(U, Z, A) \land (Z = b)$
\State \textbf{let} $[GLB_A, LUB_A] = \cons{Q', \mathcal{I}}$
\State $A_{ans} = A_{ans} \cup (b, [GLB_A, LUB_A])$
\EndFor
\State \Return $A_{ans}$
\EndProcedure
\end{algorithmic}\label{alg:grouping}
\end{algorithm}

As noted earlier, the bags of witnesses used in the preceding reductions capture the provenance of unions of conjunctive queries in the provenance polynomials model of \cite{DBLP:conf/pods/GreenKT07,DBLP:journals/sigmod/KarvounarakisG12}. In \cite{DBLP:conf/pods/AmsterdamerDT11}, it was shown that a stronger provenance model is needed to express the provenance of aggregation queries, a model that uses a tensor product combining annotations with values. A future direction of research is to investigate whether this stronger provenance model can be used to produce more direct reductions of the range consistent answers to SAT.

\section{Beyond Key Constraints}\label{beyond-key-constraints}
Key constraints and functional dependencies are important special cases of \emph{denial constraints} (DCs), which are expressible by first-order formulas of the form
$\forall x_1, ..., x_n \neg (\varphi(x_1, ..., x_n) \land \psi(x_1, ..., x_n)),$
or, equivalently,
$\forall x_1, ..., x_n  (\varphi(x_1, ..., x_n) \rightarrow \neg \psi(x_1, ..., x_n)),$
where $\varphi(x_1, ..., x_n)$ is a conjunction of atomic formulas and $\psi(x_1, ..., x_n)$ is a conjunction of  expressions of the form $(x_i\; \mbox{op}\; x_j)$ with each $\mbox{op}$  a built-in predicate, such as $=, \neq, <, >, \leq, \geq$.
In words, a denial constraint prohibits a set of tuples that satisfy certain conditions from appearing together in a database instance. If $\Sigma$ is a fixed finite set of denial constraints and $Q$ is an aggregation query without grouping, then the following problem is in coNP: given a database instance $\mathcal{I}$ and a number $t$, is $t$ the \textit{lub}-answer \textit{(or the \textit{glb}-answer)} in $\cons{Q,\mathcal{I}}$ w.r.t.\ $\Sigma$? 
This is so because to check that $t$ is not the \textit{lub}-answer \textit{(or the \textit{glb}-answer)}, we guess a repairs $\mathcal{J}$ of $\mathcal{I}$ and verify that $t > Q(\mathcal{J})$ \textit{(or $t > Q(\mathcal{J})$)}. 
In all preceding reductions, the $\alpha$-clauses capture the inconsistency in the database arisen due to the key violations to enforce every satisfying assignment to uniquely correspond to a repair of the initial inconsistent database instance. Importantly, the $\alpha$-clauses are independent of the input query. In what follows, we provide a way to construct clauses to capture the inconsistency arising due to the violations of denial constraints. Thus, replacing the $\alpha$-clauses in the reductions from Section \ref{main-body} by the ones provided below allows us to compute consistent answers over databases with a fixed finite set of arbitrary denial constraints. The reduction relies on the notions of \textit{minimal violations} and \textit{near-violations} to the set of denial constraints that we introduce next.

$\bullet$~ 
Assume that $\Sigma$ is a set of denial constraints,  $\mathcal{I}$ is an $\mathcal{R}$-instance, and  $S$ is a sub-instance of $\mathcal{I}$. We say that $S$ is a \textit{minimal violation} to $\Sigma$, if $S \not\models \Sigma$  and for every set $S' \subset S$, we have that $S' \models \Sigma$.

\smallskip
$\bullet$~ Let  $\Sigma$ be a set of denial constraints,  $\mathcal{I}$  an $\mathcal{R}$-instance,  $S$  a sub-instance of $\mathcal{I}$, and  $f$  a fact of $\mathcal{I}$. We say that $S$ is a \textit{near-violation} w.r.t.\ $\Sigma$ and $f$ if $S \models \Sigma$ and $S \cup \{f\}$ is a minimal violation to $\Sigma$. As a special case, if $\{f\}$ itself is a minimal violation to $\Sigma$,  we say that there is exactly one near-violation w.r.t. $f$, and it is the singleton  $\{f_{true}\}$, where $f_{true}$ is an auxiliary fact.



Let $\mathcal{R}$ be a database schema, $\Sigma$ be a fixed finite set of denial constraints on $\mathcal R$, $Q$ be an aggregation query without grouping, and $\mathcal{I}$ be an $\mathcal{R}$-instance.

\begin{reduction}\label{reduction3}
Given an $\mathcal R$-instance $\mathcal{I}$, compute the  sets:
\begin{enumerate}
\item $\mathcal{V}$: the set of minimal violations to $\Sigma$ on $\mathcal{I}$.
\item $\mathcal{N}^i$: the set of near-violations to $\Sigma$, on $\mathcal{I}$, w.r.t. each fact $f_i \in I$.
\end{enumerate}

\noindent For each fact $f_i$ of $\mathcal{I}$, introduce a boolean variable $x_i$, $1\leq i\leq n$. For the auxiliary fact $f_{true}$, introduce a constant  $x_{true} = true$, and for each $N^i_j \in \mathcal{N}^i$, introduce a boolean variable $p^i_j$.
\begin{enumerate}
\item For each $V_j \in \mathcal{V}$, construct a clause $\alpha_j = \underset{f_i \in V_j}{\lor} \neg x_i$.
\item For each $f_i \in I$, construct a clause $\gamma_i = x_i \lor \bigg(\underset{N^i_j \in \mathcal{N}^i}{\lor}p^i_j\bigg)$.
\item For each variable $p^i_j$, construct an expression $\theta^i_j = p^i_j \leftrightarrow \bigg(\underset{f_d \in N^i_j}{\land} x_d\bigg)$.
\item Construct the following boolean formula $\phi$:
$${\phi = \bigg(\overset{|\mathcal{V}|}{\underset{i=1}{\land}}\alpha_i\bigg)\land \bigg(\overset{|I|}{\underset{i=1}{\land}}\bigg(\Big(\overset{|\mathcal{N}^i|}{\underset{j=1}{\land}}\theta^i_j\Big)\land \gamma_i\bigg)\bigg)}$$
\end{enumerate}
\end{reduction}

\begin{proposition}\label{prop5} The boolean formula $\phi$ constructed using Reduction \ref{reduction3}  can be transformed to an equivalent CNF-formula $\phi$ whose size is polynomial in the size of $\mathcal{I}$. The satisfying assignments to $\phi$ and the repairs of $\mathcal{I}$ w.r.t.\ $\Sigma$ are in one-to-one correspondence.
\end{proposition}

\begin{proof}
Let $n$ be the number of facts of $\mathcal{I}$. Let $d_1$ be the smallest number such that there exists no denial constraint in $\Sigma$ whose number of database atoms is bigger than $d_1$. Also, let $d_2$ be the smallest number such that there exists no conjunctive query in $Q$ whose number of database atoms is bigger than $d_2$. Since $\Sigma$ and $Q$ are not part of the input to \cons{$Q$}, the quantities $d_1$ and $d_2$ are fixed constants.  We also have  that $|\mathcal{V}| \leq n^{d_1}$, $|\mathcal{N}^i| \leq n^{d_1}$ for $1 \leq i \leq n$, $|\mathcal{A}| \leq n^{d_2}$, and $|\mathcal{W}^l| \leq n^{d_2}$ for $1 \leq l \leq |\mathcal{A}|$. The number of $x$-, $y$-, and $p$-variables in $\phi'$ is therefore bounded by $n$, $n^{d_1 + 1}$, and $n^{d_2}$, respectively. The formula $\phi'$ contains as many $\alpha$-clauses as $|\mathcal{V}|$, and none of the $\alpha$-clause's length exceeds $n$. Similarly, there are at most $n^{d_2}$ $\beta$-clauses, and none of their lengths exceeds $d_2+1$. The number of $\gamma$-clauses is precisely $n$, and each $\gamma$-clause is at most $n^{d_1+1} + 1$ literals long. There are as many $\theta$-expressions as there are $y$-variables. Every $\theta$-expression is of the form $y \leftrightarrow (x_1 \land ... \land x_d)$, where $d$ is a constant obtained from the number of facts in the corresponding near-violation. Each $\theta$-expression can be equivalently written in a constant number of CNF-clauses as $((\neg y \lor x_1) \land ... \land (\neg y \lor x_d)) \land (\neg x_1 \lor ... \lor \neg x_d \lor y)$, in which the length each clause is constant. Thus, one can transform $\phi'$ into an equivalent CNF-formula $\phi$ with size polynomial in the size of $\mathcal{I}$.

For the second part of Proposition \ref{prop5}, consider a satisfying assignment $s$ to $\phi$ and construct a database instance $\mathcal{J}$ such that $f_i \in \mathcal{J}$ if and only if $s(x_i) = 1$. The $\alpha$-clauses assert that no minimal violation to $\Sigma$ is present in $\mathcal{J}$, i.e., $\mathcal{J}$ is a consistent subset of $\mathcal{I}$. The $\gamma$-clauses and the $\theta$-expressions encode the condition that, for every fact $f \in \mathcal{I}$, either $f \in \mathcal{J}$ or at least one near-violation w.r.t.\ $\Sigma$ and $f$ is in $\mathcal{J}$, making sure that $\mathcal{J}$ is indeed a repair of $\mathcal{I}$. In the other direction, one can construct a satisfying assignment $s$ to $\phi$ from a repair $\mathcal{J}$ of $\mathcal{I}$ by setting $s(x_i) = 1$ if and only if $f_i \in \mathcal{J}$.
\end{proof}
\section{Experimental Evaluation}
\label{sec:experiments}
We evaluate the performance of AggCAvSAT over both synthetic and real-world databases. The first set of experiments includes a comparison of AggCAvSAT with an existing SQL-rewriting-based CQA system, namely, ConQuer, over synthetically generated TPC-H databases having one key constraint per relation. This set of experiments is divided into two parts, based on the method used to generate the inconsistent database instances. In the first part, we use the \texttt{DBGen} tool from TPC-H and artificially inject inconsistencies in the generated data;  in the second part, we employ the PDBench inconsistent database generator from MayBMS \cite{Antova08} (see Section \ref{sec:datasets} for details). Next, we assess the scalability of AggCAvSAT by varying the size of the database and the amount of inconsistency present in it\ignore{ using our database instance generator}. Lastly, to evaluate the performance of the reductions from Section \ref{beyond-key-constraints}, we use a real-world Medigap \cite{medigap} dataset that has three functional dependencies and one denial constraint. All experiments were carried out on a machine running on Intel Core i7 2.7 GHz, 64 bit Ubuntu 16.04, with 8GB of RAM. We used Microsoft SQL Server 2017 as an underlying DBMS, and MaxHS v3.2 solver \cite{Davies2011} for solving the \wpmaxsat{} instances. The AggCAvSAT system is implemented in Java 9.04 and its code is open-sourced at a GitHub repository \url{https://github.com/uccross/cavsat} via a BSD-style license. Various features of AggCAvSAT, including its graphical user interface, are presented in a short paper in the demonstration track of the 2021 SIGMOD conference \cite{DixitK2021}.

\subsection{Experiments with TPC-H Data and Queries}
\subsubsection{Datasets}\label{sec:datasets}  For the first part of the experiments, the data is generated using the \texttt{DBGen} data generation tool from the TPC-H Consortium. The TPC-H schema comes with exactly one key constraint per relation, which was ideal for comparing AggCAvSAT against ConQuer \cite{Fuxman05, Fuxman07} (the only existing system for computing the range consistent answers to aggregation queries), because ConQuer does not support more than one key constraint per relation or  classes of integrity constraints broader than keys. The \texttt{DBGen} tool generates consistent data, so we artificially injected inconsistency by updating the key attributes of randomly selected tuples from the data with the values taken from existing tuples of the same relation. The sizes of the key-equal groups that violate the key constraints were uniformly distributed between two and seven. The database instances were generated in such a way that every repair had the specified size. We experimented with varied degrees of inconsistency, ranging from 5\% up to 35\% of the tuples violating a key constraint, 
and with a variety of repair sizes, starting from 500 MB (4.3 million tuples) up to 5 GB (44 million tuples).  For the second part, we employed the PDBench database generator from MayBMS \cite{Antova08} to generate four inconsistent database instances with varying degrees of inconsistency (see Table \ref{pdbench-tpch}). In all four instances, the data is generated in such a way that every repair is of size 1 GB.

\begin{table}[t]
\caption{Percentage of inconsistency in the TPC-H database instances generated using PDBench} \label{pdbench-tpch}
\begin{tabular}{|l|c|c|c|c|}\hline
&\multicolumn{4}{c|}{\textbf{Inconsistency}}\\\hline
\textbf{Table}&\textbf{Inst.\ 1}&\textbf{Inst.\ 2}&\textbf{Inst.\ 3}&\textbf{Inst.\ 4}\\\hline
\texttt{CUSTOMER}&4.42\%&8.5\%&16.14\%&29.49\%\\
\texttt{LINEITEM}&6.36\%&12.09\%&22.53\%&39.82\%\\
\texttt{NATION}&7.69\%&0\%&7.69\%&7.69\%\\
\texttt{ORDERS}&3.51\%&6.77\%&12.87\%&23.9\%\\
\texttt{PART}&4.93\%&9.33\%&17.66\%&32.16\%\\
\texttt{PARTSUPP}&1.53\%&2.96\%&5.77\%&11.13\%\\
\texttt{REGION}&0\%&0\%&0\%&0\%\\
\texttt{SUPPLIER}&3.69\%&7.44\%&14.11\%&26.51\%\\\hline
\textbf{Overall}&5.36\%&10.25\%&19.29\%&34.72\%\\\hline\hline
&\multicolumn{4}{c|}{\textbf{Database size and Repair size (in GB)}}\\\hline
&1.04 \& 1.00&1.07 \& 1.01&1.14 \& 1.02&1.3 \& 1.02\\\hline\hline
&\multicolumn{4}{c|}{\textbf{Size of the Largest Key-equal Groups}}\\\hline
&8 tuples &16 tuples&16 tuples&32 tuples\\\hline
\end{tabular}
\end{table}

\subsubsection{Queries} The standard TPC-H specification comes with 22 queries (constructed using the \texttt{QGen} tool). Here, we focus on queries 1, 3, 4, 5, 6, 10, 12, 14, and 19; the other 13 queries have features such as nested subqueries, left outer joins, and negation that are beyond the aggregation queries defined in Section \ref{sec-aggregation}. In Section \ref{sec-exp-w/o-grouping}, we describe our results for queries without grouping. Since six out of the nine queries under consideration contained the \texttt{GROUP BY} clause, we removed  it and added appropriate conditions in the \texttt{WHERE} clause based on the original grouping attributes to obtain queries without grouping. We refer to these queries as $Q'_1, Q'_3,\cdots, Q'_{19}$. The definitions of these queries are given in Table \ref{tpch-data-queries}.

\begin{table}[ht]
\scriptsize
\caption{TPC-H-inspired aggregation queries w/o grouping}
\centering
\begin{tabular}{|p{0.03 \linewidth}|p{0.9\linewidth}|} \hline
\#&Query\\ \hline
$Q'_1$&\texttt{SELECT SUM(LINEITEM.L\_QUANTITY) FROM LINEITEM WHERE LINEITEM.L\_SHIPDATE <= dateadd(dd, -90, cast(`1998-12-01' as datetime)) AND LINEITEM.L\_RETURNFLAG = `N' AND LINEITEM.L\_LINESTATUS = `F'}\\
$Q'_3$&\texttt{SELECT SUM(LINEITEM.L\_EXTENDEDPRICE*(1-LINEITEM.L\_DISCOUNT)) FROM CUSTOMER, ORDERS, LINEITEM WHERE CUSTOMER.C\_MKTSEGMENT = 'BUILDING' AND CUSTOMER.C\_CUSTKEY = ORDERS.O\_CUSTKEY AND LINEITEM.L\_ORDERKEY = ORDERS.O\_ORDERKEY AND ORDERS.O\_ORDERDATE < `1995-03-15' AND LINEITEM.L\_SHIPDATE > `1995-03-15' AND LINEITEM.L\_ORDERKEY = 988226 AND ORDERS.O\_ORDERDATE = `1995-02-01' AND ORDERS.O\_SHIPPRIORITY = 0}\\
$Q'_4$&\texttt{SELECT COUNT(*) FROM ORDERS WHERE ORDERS.O\_ORDERDATE >= `1993-07-01' AND ORDERS.O\_ORDERDATE < dateadd(mm,3, cast(`1993-07-01' as datetime)) AND ORDERS.O\_ORDERPRIORITY = `1-URGENT'}\\
$Q'_5$&\texttt{SELECT SUM(LINEITEM.L\_EXTENDEDPRICE*(1-LINEITEM.L\_DISCOUNT)) FROM CUSTOMER, ORDERS, LINEITEM, SUPPLIER, NATION, REGION WHERE CUSTOMER.C\_CUSTKEY = ORDERS.O\_CUSTKEY AND LINEITEM.L\_ORDERKEY = ORDERS.O\_ORDERKEY AND LINEITEM.L\_SUPPKEY = SUPPLIER.S\_SUPPKEY AND CUSTOMER.C\_NATIONKEY = SUPPLIER.S\_NATIONKEY AND SUPPLIER.S\_NATIONKEY = NATION.N\_NATIONKEY AND NATION.N\_REGIONKEY = REGION.R\_REGIONKEY AND REGION.R\_NAME = `ASIA' AND ORDERS.O\_ORDERDATE >= `1994-01-01' AND ORDERS.O\_ORDERDATE < DATEADD(YY, 1, cast(`1994-01-01' as datetime)) AND NATION.N\_NAME = `INDIA'}\\
$Q'_6$&\texttt{SELECT SUM(LINEITEM.L\_EXTENDEDPRICE*LINEITEM.L\_DISCOUNT) FROM LINEITEM WHERE LINEITEM.L\_SHIPDATE >= `1994-01-01' AND LINEITEM.L\_SHIPDATE < dateadd(yy, 1, cast(`1994-01-01' as datetime)) AND LINEITEM.L\_DISCOUNT BETWEEN .06 - 0.01 AND .06 + 0.01 AND LINEITEM.L\_QUANTITY < 24}\\
$Q'_{10}$&\texttt{SELECT SUM(LINEITEM.L\_EXTENDEDPRICE*(1-LINEITEM.L\_DISCOUNT)) FROM CUSTOMER, ORDERS, LINEITEM, NATION WHERE CUSTOMER.C\_CUSTKEY = ORDERS.O\_CUSTKEY AND LINEITEM.L\_ORDERKEY = ORDERS.O\_ORDERKEY AND ORDERS.O\_ORDERDATE >= `1993-10-01' AND ORDERS.O\_ORDERDATE < dateadd(mm, 3, cast(`1993-10-01' as datetime)) AND LINEITEM.L\_RETURNFLAG = `R' AND CUSTOMER.C\_NATIONKEY = NATION.N\_NATIONKEY AND CUSTOMER.C\_CUSTKEY = 77296 AND CUSTOMER.C\_NAME = `Customer\#000077296' AND CUSTOMER.C\_ACCTBAL = 1250.65 AND CUSTOMER.C\_PHONE = `12-248-307-9719' AND NATION.N\_NAME = `BRAZIL'}\\
$Q'_{12}$&\texttt{SELECT COUNT(*) FROM ORDERS, LINEITEM WHERE ORDERS.O\_ORDERKEY = LINEITEM.L\_ORDERKEY AND LINEITEM.L\_SHIPMODE = `MAIL' AND (ORDERS.O\_ORDERPRIORITY = `1-URGENT' OR ORDERS.O\_ORDERPRIORITY = `2-HIGH') AND LINEITEM.L\_COMMITDATE < LINEITEM.L\_RECEIPTDATE AND LINEITEM.L\_SHIPDATE < LINEITEM.L\_COMMITDATE AND LINEITEM.L\_RECEIPTDATE >= `1994-01-01' AND LINEITEM.L\_RECEIPTDATE < dateadd(mm, 1, cast(`1995-09-01' as datetime))}\\
$Q'_{14}$&\texttt{SELECT SUM(LINEITEM.L\_EXTENDEDPRICE*(1-LINEITEM.L\_DISCOUNT)) FROM LINEITEM, PART WHERE LINEITEM.L\_PARTKEY = PART.P\_PARTKEY AND LINEITEM.L\_SHIPDATE >= `1995-09-01' AND LINEITEM.L\_SHIPDATE < dateadd(mm, 1, `1995-09-01') AND PART.P\_TYPE LIKE `PROMO\%\%'}\\
$Q'_{19}$&\texttt{SELECT SUM(LINEITEM.L\_EXTENDEDPRICE*(1-LINEITEM.L\_DISCOUNT)) FROM LINEITEM, PART WHERE (PART.P\_PARTKEY = LINEITEM.L\_PARTKEY AND PART.P\_BRAND = `Brand\#12' AND PART.P\_CONTAINER IN (`SM CASE', `SM BOX', `SM PACK', `SM PKG') AND LINEITEM.L\_QUANTITY >= 1 AND LINEITEM.L\_QUANTITY <= 1 + 10 AND PART.P\_SIZE BETWEEN 1 AND 5 AND LINEITEM.L\_SHIPMODE IN (`AIR', `AIR REG') AND LINEITEM.L\_SHIPINSTRUCT = `DELIVER IN PERSON') OR (PART.P\_PARTKEY = LINEITEM.L\_PARTKEY AND PART.P\_BRAND =`Brand\#23' AND PART.P\_CONTAINER IN (`MED BAG', `MED BOX', `MED PKG', `MED PACK') AND LINEITEM.L\_QUANTITY >= 10 AND LINEITEM.L\_QUANTITY <= 10 + 10 AND PART.P\_SIZE BETWEEN 1 AND 10 AND LINEITEM.L\_SHIPMODE IN (`AIR', `AIR REG') AND LINEITEM.L\_SHIPINSTRUCT = `DELIVER IN PERSON') OR (PART.P\_PARTKEY = LINEITEM.L\_PARTKEY AND PART.P\_BRAND = `Brand\#34' AND PART.P\_CONTAINER IN ( `LG CASE', `LG BOX', `LG PACK', `LG PKG') AND LINEITEM.L\_QUANTITY >= 20 AND LINEITEM.L\_QUANTITY <= 20 + 10 AND PART.P\_SIZE BETWEEN 1 AND 15 AND LINEITEM.L\_SHIPMODE IN (`AIR', `AIR REG') AND LINEITEM.L\_SHIPINSTRUCT = `DELIVER IN PERSON')}\\
\hline\end{tabular}
\label{tpch-data-queries}
\end{table}

\subsubsection{Results on Queries without Grouping}\label{sec-exp-w/o-grouping}
In the first set of experiments, we computed the range consistent answers of the TPC-H-inspired aggregation queries without grouping via WPMaxSAT solving over a database instance with 10\% inconsistency and having repairs of size 1 GB (8 million tuples). Figure \ref{wo-grouping} shows that much of the evaluation time used by AggCAvSAT
is consumed in encoding the CQA instance into a WPMaxSAT instance, while the solver comparatively takes less time to compute the optimal solution. Note that, $Q'_5$ is not in the class $C_\textit{aggforest}$ and thus ConQuer cannot compute its range consistent answers. AggCAvSAT performs better than ConQuer on seven out of the remaining eight queries.
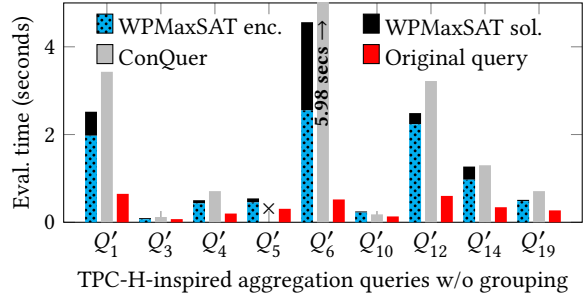
\begin{figure}[ht]
\begin{tikzpicture}[every axis/.style={height = 4.5cm, width=\linewidth,ybar stacked, xlabel= TPC-H-inspired aggregation queries w/o grouping, xtick=data, xticklabels={$Q'_1$,$Q'_3$,$Q'_4$,$Q'_5$,$Q'_6$,$Q'_{10}$,$Q'_{12}$,$Q'_{14}$,$Q'_{19}$}, legend cell align = {left}, bar width=4.5pt, ymin=0,ymax=5, y label style={at={(0.1,0.5)}}, x label style={at={(0.5,0)}}, ylabel=Eval.\ time (seconds),legend columns=2, legend style={draw=none, fill=none,/tikz/every even column/.append style={column sep=1cm}, at={(0.95, 0.98)}}}]
\begin{axis}[bar shift=-6pt]
\addplot[fill=cyan, draw=none,postaction={pattern=crosshatch dots, samples=20}] coordinates {(1,1.973)(2,0.07)(3,0.432)(4,0.453)(5,2.55)(6,0.221)(7,2.235)(8,0.97)(9,0.481)};\label{encoding}
\addplot[fill=black, draw=none] coordinates {(1,0.547)(2,0.025)(3,0.071)(4,0.091)(5,2.011)(6,0.026)(7,0.254)(8,0.3)(9,0.03)};\label{solving}
\end{axis}
\begin{axis}[hide axis]
\addplot[fill=lightgray, draw=none] coordinates {(1,3.43)(2,0.12)(3,0.71)(4,0)(5,5.972)(6,0.18)(7,3.22)(8,1.3)(9,0.71)}
node[above] at (axis cs:4,0){\small $\times$}node[rotate=90, yshift=0.02cm] at (axis cs:5,3.4) {\textbf{\small 5.98 secs $\rightarrow$}}; \label{conquer}
\end{axis}
\begin{axis}[bar shift=6pt, hide axis]
\addlegendimage{/pgfplots/refstyle=encoding}\addlegendentry{WPMaxSAT enc.}
\addlegendimage{/pgfplots/refstyle=solving}\addlegendentry{WPMaxSAT sol.}
\addlegendimage{/pgfplots/refstyle=conquer}\addlegendentry{ConQuer}
\addplot[fill=red, draw=none] coordinates {(1,0.648)(2,0.072)(3,0.199)(4,0.307)(5,0.52)(6,0.133)(7,0.601)(8,0.343)(9,0.271)};
\addlegendentry{Original query}
\end{axis}
\end{tikzpicture}
\caption{AggCAvSAT vs.\ ConQuer on TPC-H data generated using the \texttt{DBGen}-based tool (10\% inconsistency, 1 GB repairs)}
\label{wo-grouping}
\end{figure}

Next, we compared the performance of AggCAvSAT and ConQuer on database instances generated using PDBench. Figure \ref{wo-grouping-pdbench} shows that  AggCAvSAT performs better than ConQuer on PDBench instances with low inconsistency. As the inconsistency increases, the \wpmaxsat{} solver requires considerably long time to compute the optimal solutions (especially for $Q'_6$, $Q'_{12}$, and $Q'_{14}$). One  reason  is that the sizes of key-equal groups in PDBench instances with higher inconsistency percentage are large, which translates into clauses of large sizes in the \wpmaxsat{} instances, hence the solver works hard to solve them. Also, K\"{u}gel's reduction \cite{Kuegel12} from \wpminsat{} to \wpmaxsat{} significantly increases the size of the CNF formula, resulting in higher time for the \textit{lub}-answers to the queries.

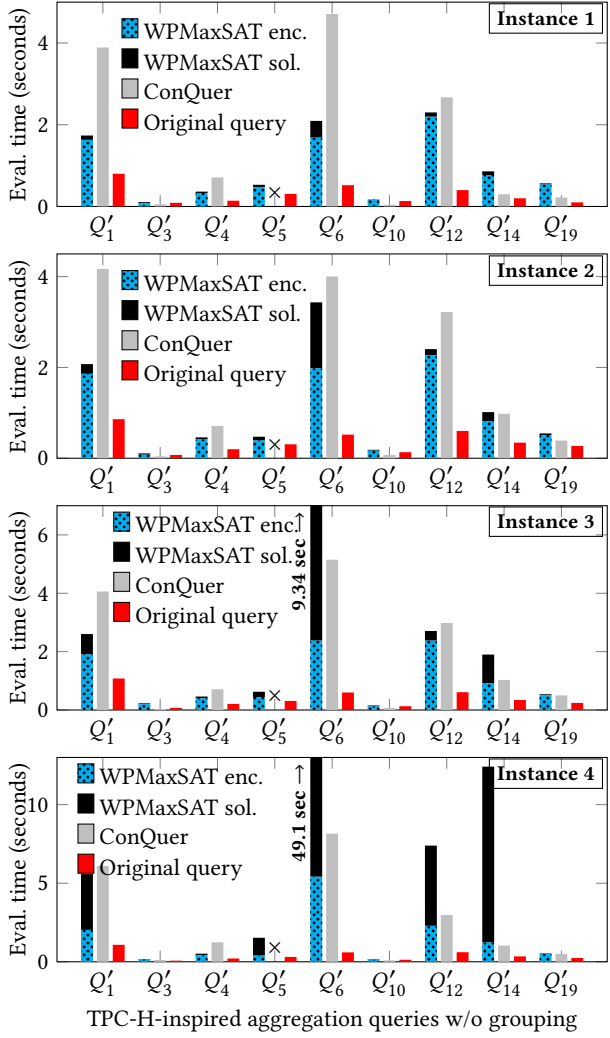
\begin{figure}[ht]
\centering
\begin{subfigure}[b]{0.5\textwidth}
\begin{tikzpicture}[every axis/.style={height=4.3cm,width=\linewidth,ybar stacked, xtick=data, xticklabels={$Q'_1$,$Q'_3$,$Q'_4$,$Q'_5$,$Q'_6$,$Q'_{10}$,$Q'_{12}$,$Q'_{14}$,$Q'_{19}$}, legend cell align = {left}, bar width=4.5pt, ymin=0,ymax=5, y label style={at={(0.1,0.5)}}, x label style={at={(0.5,0)}}, ylabel=Eval.\ time (seconds),legend style={draw=none, fill=none, at={(0.48, 0.95)}}}]
\begin{axis}[bar shift=-6pt]
\addplot[fill=cyan, draw=none,postaction={pattern=crosshatch dots, samples=20}] coordinates
{(1,1.63)(2,0.08)(3,0.329)(4,0.478)(5,1.683)(6,0.176)(7,2.196)(8,0.768)(9,0.537)};\label{encoding1}
\addplot[fill=black, draw=none] coordinates{(1,0.106)(2,0.026)(3,0.034)(4,0.052)(5,0.411)(6,0)(7,0.105)(8,0.091)(9,0.028)};\label{solving1}
\end{axis}
\begin{axis}[hide axis]
\addplot[fill=lightgray, draw=none] coordinates {(1,3.89)(2,0.05)(3,0.71)(4,0)(5,4.71)(6,0.04)(7,2.67)(8,0.3)(9,0.22)}
node[above] at (axis cs:4,0){\small $\times$}node[yshift=0.37cm, draw=black] at (axis cs:8.73,3.92){\textbf{\small Instance 1}};\label{conquer1}
\end{axis}
\begin{axis}[bar shift=6pt, hide axis]
\addlegendimage{/pgfplots/refstyle=encoding}\addlegendentry{WPMaxSAT enc.}
\addlegendimage{/pgfplots/refstyle=solving}\addlegendentry{WPMaxSAT sol.}
\addlegendimage{/pgfplots/refstyle=conquer}\addlegendentry{ConQuer}
\addplot[fill=red, draw=none] coordinates {(1,0.8)(2,0.09)(3,0.14)(4,0.307)(5,0.52)(6,0.133)(7,0.4)(8,0.2)(9,0.1)};
\addlegendentry{Original query}
\end{axis}
\end{tikzpicture}
\end{subfigure}
\begin{subfigure}[b]{0.5\textwidth}
\begin{tikzpicture}[every axis/.style={height=4.3cm,width=\linewidth,ybar stacked, xtick=data, xticklabels={$Q'_1$,$Q'_3$,$Q'_4$,$Q'_5$,$Q'_6$,$Q'_{10}$,$Q'_{12}$,$Q'_{14}$,$Q'_{19}$}, legend cell align = {left}, bar width=4.5pt, ymin=0,ymax=4.5, y label style={at={(0.1,0.5)}}, x label style={at={(0.5,0)}}, ylabel=Eval.\ time (seconds),legend style={draw=none, fill=none, at={(0.48, 0.95)}}}]
\begin{axis}[bar shift=-6pt]
\addplot[fill=cyan, draw=none,postaction={pattern=crosshatch dots, samples=20}] coordinates
{(1,1.859)(2,0.078)(3,0.414)(4,0.39)(5,1.979)(6,0.166)(7,2.272)(8,0.818)(9,0.509)};\label{encoding1}
\addplot[fill=black, draw=none] coordinates{(1,0.214)(2,0.029)(3,0.045)(4,0.084)(5,1.454)(6,0.02)(7,0.131)(8,0.198)(9,0.037)};\label{solving1}
\end{axis}
\begin{axis}[hide axis]
\addplot[fill=lightgray, draw=none] coordinates {(1,4.17)(2,0.042)(3,0.71)(4,0)(5,4.002)(6,0.073)(7,3.22)(8,0.98)(9,0.39)}
node[above] at (axis cs:4,0){\small $\times$}node[yshift=0.37cm, draw=black] at (axis cs:8.73,3.52){\textbf{\small Instance 2}}; \label{conquer1}
\end{axis}
\begin{axis}[bar shift=6pt, hide axis]
\addlegendimage{/pgfplots/refstyle=encoding}\addlegendentry{WPMaxSAT enc.}
\addlegendimage{/pgfplots/refstyle=solving}\addlegendentry{WPMaxSAT sol.}
\addlegendimage{/pgfplots/refstyle=conquer}\addlegendentry{ConQuer}
\addplot[fill=red, draw=none] coordinates {(1,0.86)(2,0.072)(3,0.199)(4,0.307)(5,0.52)(6,0.133)(7,0.601)(8,0.343)(9,0.271)};
\addlegendentry{Original query}
\end{axis}
\end{tikzpicture}
\end{subfigure}
\begin{subfigure}[b]{0.5\textwidth}
\begin{tikzpicture}[every axis/.style={height=4.3cm,width=\linewidth,ybar stacked, xtick=data, xticklabels={$Q'_1$,$Q'_3$,$Q'_4$,$Q'_5$,$Q'_6$,$Q'_{10}$,$Q'_{12}$,$Q'_{14}$,$Q'_{19}$}, legend cell align = {left}, bar width=4.5pt, ymin=0,ymax=7, y label style={at={(0.1,0.5)}}, x label style={at={(0.5,0)}}, ylabel=Eval.\ time (seconds),legend style={draw=none, fill=none, at={(0.465, 1)}}}]
\begin{axis}[bar shift=-6pt]
\addplot[fill=cyan, draw=none,postaction={pattern=crosshatch dots, samples=20}] coordinates{(1,1.904)(2,0.2)(3,0.397)(4,0.446)(5,2.376)(6,0.128)(7,2.376)(8,0.914)(9,0.502)};\label{encoding1}
\addplot[fill=black, draw=none] coordinates{(1,0.698)(2,0.029)(3,0.065)(4,0.18)(5,7.065)(6,0.028)(7,0.332)(8,0.986)(9,0.038)}node[rotate=90, yshift=0.37cm] at (axis cs:4.9,5.1) {\textbf{\small 9.34 sec $\rightarrow$}};\label{solving1}
\end{axis}
\begin{axis}[hide axis]
\addplot[fill=lightgray, draw=none] coordinates {(1,4.06)(2,0.042)(3,0.71)(4,0)(5,5.15)(6,0.073)(7,2.98)(8,1.03)(9,0.5)}
node[above] at (axis cs:4,0){\small $\times$}node[yshift=0.37cm, draw=black] at (axis cs:8.73,5.5){\textbf{\small Instance 3}}; \label{conquer1}
\end{axis}
\begin{axis}[bar shift=6pt, hide axis]
\addlegendimage{/pgfplots/refstyle=encoding}\addlegendentry{WPMaxSAT enc.}
\addlegendimage{/pgfplots/refstyle=solving}\addlegendentry{WPMaxSAT sol.}
\addlegendimage{/pgfplots/refstyle=conquer}\addlegendentry{ConQuer}
\addplot[fill=red, draw=none] coordinates {(1,1.08)(2,0.072)(3,0.21)(4,0.31)(5,0.6)(6,0.13)(7,0.61)(8,0.343)(9,0.243)};
\addlegendentry{Original query}
\end{axis}
\end{tikzpicture}
\end{subfigure}
\begin{subfigure}[b]{0.5\textwidth}
\begin{tikzpicture}[every axis/.style={height=4.3cm,width=\linewidth,ybar stacked, xlabel= TPC-H-inspired aggregation queries w/o grouping, xtick=data, xticklabels={$Q'_1$,$Q'_3$,$Q'_4$,$Q'_5$,$Q'_6$,$Q'_{10}$,$Q'_{12}$,$Q'_{14}$,$Q'_{19}$}, legend cell align = {left}, bar width=4.5pt, ymin=0,ymax=13, y label style={at={(0.1,0.5)}}, x label style={at={(0.5,0)}}, ylabel=Eval.\ time (seconds),legend style={draw=none, fill=none, at={(0.4, 1)}}}]
\begin{axis}[bar shift=-6pt]
\addplot[fill=cyan, draw=none,postaction={pattern=crosshatch dots, samples=20}] coordinates{(1,2.029)(2,0.143)(3,0.414)(4,0.394)(5,5.412)(6,0.132)(7,2.288)(8,1.228)(9,0.489)};\label{encoding1}
\addplot[fill=black, draw=none] coordinates{(1,3.55)(2,0.03)(3,0.102)(4,1.134)(5,46.297)(6,0.03)(7,5.104)(8,11.188)(9,0.04)}node[rotate=90, yshift=0.37cm] at (axis cs:4.9,9.5) {\textbf{\small 49.1 sec $\rightarrow$}};\label{solving1}
\end{axis}
\begin{axis}[hide axis]
\addplot[fill=lightgray, draw=none] coordinates {(1,6.1)(2,0.12)(3,1.24)(4,0)(5,8.15)(6,0.09)(7,2.98)(8,1.03)(9,0.5)}
node[above] at (axis cs:4,0){\small $\times$}node[yshift=0.37cm, draw=black] at (axis cs:8.73,10.2){\textbf{\small Instance 4}}; \label{conquer1}
\end{axis}
\begin{axis}[bar shift=6pt, hide axis]
\addlegendimage{/pgfplots/refstyle=encoding}\addlegendentry{WPMaxSAT enc.}
\addlegendimage{/pgfplots/refstyle=solving}\addlegendentry{WPMaxSAT sol.}
\addlegendimage{/pgfplots/refstyle=conquer}\addlegendentry{ConQuer}
\addplot[fill=red, draw=none] coordinates {(1,1.08)(2,0.072)(3,0.21)(4,0.31)(5,0.6)(6,0.13)(7,0.61)(8,0.343)(9,0.243)};
\addlegendentry{Original query}
\end{axis}
\end{tikzpicture}
\end{subfigure}
\caption{AggCAvSAT vs.\ ConQuer on PDBench instances}
\label{wo-grouping-pdbench}
\end{figure}

\begin{table}[h]
\caption{Average size of the CNF formulas for $Q'_1$, $Q'_6$, and $Q'_{14}$}
\label{formula-size-varying-datasize}
\begin{subtable}{0.25\textwidth}
\centering
\begin{tabular}{c|c|c|c|c|} \cline{2-5}
{}&5\%&15\%&25\%&35\%\\ \cline{2-5}
$Q'_1$ &10.2& 34.3 & 60.6&95.3\\
$Q'_6$ &28.4& 96.2 & 175.0&271.2\\
$Q'_{14}$ &6.4& 21.1 & 40.6&62.3\\
\cline{2-5}\end{tabular}
\caption{\# of variables (in thousands)}
\label{formula-size-varying-incons-vars}
\end{subtable}%
\begin{subtable}{0.25\textwidth}
\centering
\begin{tabular}{|c|c|c|c|} \hline
5\%&15\%&25\%&35\%\\ \hline
27.6& 92.2 & 163.6&258.1\\
76.8& 259.9 & 472.9&734.0\\
15.6& 51.9 & 101.0&156.6\\
\hline\end{tabular}
\caption{\# of clauses (in thousands)}
\label{formula-size-varying-incons-clauses}
\end{subtable}\\%
\begin{subtable}{0.25\textwidth}
\centering
\begin{tabular}{c|c|c|c|} \cline{2-4}
{}&1 GB&3 GB&5 GB\\ \cline{2-4}
$Q'_1$ &21.3&  44.1&105.6\\
$Q'_6$ &60.9& 127.13&304.4\\
$Q'_{14}$ &13.9&32.9&67.7\\
\cline{2-4}\end{tabular}
\caption{\# of variables (in thousands)}
\label{formula-size-varying-size-vars}
\end{subtable}%
\begin{subtable}{0.25\textwidth}
\centering
\begin{tabular}{|c|c|c|} \hline
1 GB&3 GB&5 GB\\ \hline
57.7& 104.1 & 258.8\\
165.3& 300.7 & 823.1\\
34.0& 73.7 & 166.6\\
\hline\end{tabular}
\caption{\# of clauses (in thousands)}
\label{formula-size-varying-size-clauses}
\end{subtable}%
\end{table}

Next, we varied the inconsistency in the database instances created using the \texttt{DBGen}-based data generator while keeping the size of the database repairs constant (1 GB). Figure \ref{varying-inconsistency} shows that the evaluation time of AggCAvSAT stays well under ten seconds (except for $Q'_6$), even if there is more inconsistency in the data. Tables \ref{formula-size-varying-incons-vars} and \ref{formula-size-varying-incons-clauses} show the average size of the CNF formulas for the top three queries that exhibited the largest CNF formulas. The size of the formulas grows nearly linearly as the inconsistency present in the data grows. The CNF formulas for $Q'_6$ are significantly larger than the ones corresponding to the other queries since $Q'_6$ has high selectivity and it is posed against the single largest relation \texttt{LINEITEM} which has over 8.2 million tuples in an instance with 35\% inconsistency. This also explains why AggCAvSAT takes more time for computing its range consistent answers (Figure \ref{varying-inconsistency}). \ignore{The overall evaluation time for $Q'_{12}$ is high even though its CNF formula is not very big. This is because  $Q'_{12}$ involves a join between two large tables, namely, \texttt{LINEITEM} and \texttt{ORDERS}, which results in slow computation of the minimal witnesses, therefore affecting the encoding time.} In database instances with low inconsistency, the consistent answers to the queries having low selectivity (e.g., $Q'_3$, $Q'_{10}$) are sometimes contained in the consistent part of the data, and AggCAvSAT does not need to construct a \wpmaxsat{} instance at all.

\begin{figure}[h]
\begin{tikzpicture}
  \begin{axis}[width=\linewidth, grid=major,height=4.5cm, tick label style={font=\normalsize},legend columns=3, label style={font=\normalsize}, grid style={white}, xlabel={Percentage of inconsistency}, ylabel={Eval.\ time (seconds)}, y label style={at={(0.05,0.5)}}, ymax=20, y tick label style={/pgf/number format/.cd, fixed, fixed zerofill, precision=1}, legend style={draw=none, fill=none, font=\scriptsize,at={(0.53, 1)}},legend cell align={left}, every axis plot/.append style={thick}]
  \addplot[blue, mark=*] coordinates{(5,1.941) (10,2.45) (15,2.84) (20,3.6) (25,4.53)(30,5.223) (35,6.21)};\addlegendentry{$Q'_1$}
  \addplot[cyan, mark=triangle] coordinates{(5,0.102) (10,0.085) (15,0.109) (20,0.09) (25,0.118)(30,0.088) (35,0.103)};\addlegendentry{$Q'_3$}
  \addplot[orange, mark=o] coordinates{(5,0.451) (10,0.619) (15,0.763) (20,0.819) (25,1.03)(30,1.17) (35,1.591)};\addlegendentry{$Q'_4$}
  \addplot[black, mark=x] coordinates{(5,0.45) (10,0.553) (15,0.697) (20,0.891) (25,1.197)(30,1.85) (35,2.95)};\addlegendentry{$Q'_5$}
  \addplot[red, dashed, mark=*] coordinates{(5,2.412) (10,4.48) (15,7.2) (20,10.05) (25,12.3)(30,14.82) (35,18.12)};\addlegendentry{$Q'_6$}
   \addplot[purple, mark=square] coordinates{(5,0.225) (10,0.26) (15,0.256) (20,0.208) (25,0.331)(30,0.215) (35,0.483)};\addlegendentry{$Q'_{10}$}
	\addplot[violet, mark=star] coordinates{(5,2.312) (10,2.8) (15,3.12) (20,4.234) (25,5.22)(30,5.54) (35,5.9)};\addlegendentry{$Q'_{12}$}
  \addplot[brown, mark=diamond] coordinates{(5,0.982) (10,1.32) (15,1.61) (20,2.32) (25,3.6)(30,5.98) (35,8.02)};\addlegendentry{$Q'_{14}$}
  \addplot[red, mark=x] coordinates{(5,0.57) (10,0.621) (15,0.614) (20,0.691) (25,0.629)(30,0.67) (35,0.726)};\addlegendentry{$Q'_{19}$}
\end{axis}
\end{tikzpicture}
\caption{AggCAvSAT on TPC-H data generated using the \texttt{DBGen}-based tool (varying inconsistency, 1 GB repairs)}
\label{varying-inconsistency}
\end{figure}
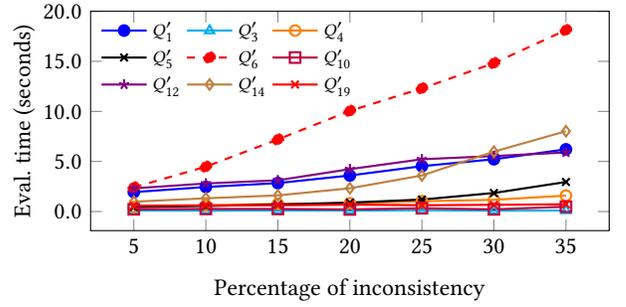

We then evaluated AggCAvSAT's scalability by increasing the sizes of the databases while keeping the inconsistency to a constant 10\%. Figure \ref{varying-datasize} shows that the evaluation time of AggCAvSAT for queries $Q'_1$, $Q'_6$, and $Q'_{12}$ increases faster than that for the other queries. This is because the queries $Q'_1$ and $Q'_6$ are posed against \texttt{LINEITEM} while $Q'_{12}$ involves a join between \texttt{LINEITEM} and \texttt{ORDERS} resulting in AggCAvSAT spending more time on computing the bags of witnesses to these queries as the size of the database grows. Table \ref{formula-size-varying-size-vars} and \ref{formula-size-varying-size-clauses} show that the size of the CNF formulas grows almost linearly w.r.t.\ the size of the database. The largest CNF formula consisted of over 304 thousand variables and 823 thousand clauses and was exhibited by $Q'_6$ on a database  of size 5 GB (47 million tuples). The low selectivity of queries  $Q'_3$, $Q'_{10}$, and $Q'_{19}$ resulted in very small CNF formulas, even on large databases.

\begin{figure}[ht]
\begin{tikzpicture}
  \begin{axis}[width=\linewidth, grid=major, height=4.5cm, tick label style={font=\normalsize},legend columns=5, label style={font=\normalsize}, grid style={white}, xlabel={Size of the database repairs (in GB)}, ylabel={Eval.\ time (seconds)}, y label style={at={(0.05,0.5)}}, ymax=45, y tick label style={/pgf/number format/.cd, fixed, fixed zerofill, precision=0}, legend style={draw=none, fill=none, font=\scriptsize,at={(0.85, 1)}},legend cell align={left}, every axis plot/.append style={thick}]
  \addplot[blue, mark=*] coordinates{(0.5,1.296)(1,2.45)(2,4.021)(3,6.272)(4,14.12)(5,23.083)};\addlegendentry{$Q_1$}
  \addplot[cyan, mark=triangle] coordinates{(0.5,0.093)(1,0.085)(2,0.055)(3,0.077)(4,0.07)(5,0.07)};\addlegendentry{$Q_3$}
  \addplot[orange, mark=o] coordinates{(0.5,0.413)(1,0.619)(2,0.921)(3,1.648)(4,1.95)(5,3.092)};\addlegendentry{$Q_4$}
  \addplot[black, mark=x] coordinates{(0.5,0.386)(1,0.553)(2,0.89)(3,1.347)(4,1.873)(5,2.258)};\addlegendentry{$Q_5$}
  \addplot[red, dashed, mark=*] coordinates{(0.5,2.02)(1,4.48)(2,7.7)(3,12.124)(4,20.1)(5,42.0)};
\addlegendentry{$Q_6$}
  \addplot[purple, mark=square] coordinates{(0.5,0.112)(1,0.26)(2,0.31)(3,0.486)(4,0.591)(5,0.53)};\addlegendentry{$Q_{10}$}
  \addplot[violet, mark=star] coordinates{(0.5,1.345)(1,2.8)(2,5.11)(3,8.078)(4,13.18)(5,20.011)};\addlegendentry{$Q_{12}$}
  \addplot[brown, mark=diamond] coordinates{(0.5,0.742)(1,1.32)(2,1.91)(3,3.603)(4,7.654)(5,11.536)};\addlegendentry{$Q_{14}$}
  \addplot[red, mark=x] coordinates{(0.5,0.336)(1,0.621)(2,0.966)(3,1.557)(4,3.68)(5,6.189)};\addlegendentry{$Q_{19}$}
\end{axis}
\end{tikzpicture}
\caption{AggCAvSAT on TPC-H data generated using the \texttt{DBGen}-based tool (varying database sizes, 10\% inconsistency)}
\label{varying-datasize}
\end{figure}
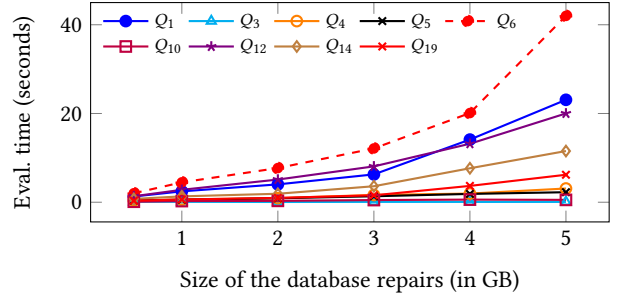
\vspace{-0.2cm}

\subsubsection{Results on Queries with Grouping}
In this set of experiments, we focus on TPC-H queries 1, 3, 4, 5, 10, and 12 (see Table \ref{tpch-data-queries-grouping}), as the queries 6, 14, and 19 did not contain grouping. We evaluated the performance of AggCAvSAT and compared it to ConQuer on a database with 10\% inconsistency w.r.t.\ primary keys (Figure \ref{w-grouping}). The repairs are of size 1 GB. AggCAvSAT computes the consistent answers to the underlying conjunctive query using the reductions from \cite{DixitK19} which are, precisely, the consistent groups in the range consistent answers to the aggregation query. For each of these groups, it computes the \textit{glb}-answer and the \textit{lub}-answer using reductions to \wpmaxsat{}.

\begin{table}[ht]
\scriptsize
\caption{TPC-H-inspired Aggregation Queries w/ Grouping}
\centering
\begin{tabular}{|p{0.017\linewidth}|p{0.78\linewidth}|p{0.08\linewidth}|} \hline
\#&Query&Operator\\ \hline
$Q_1$&\texttt{SELECT LINEITEM.L\_RETURNFLAG, LINEITEM.L\_LINESTATUS, SUM(LINEITEM.L\_QUANTITY)
FROM LINEITEM WHERE LINEITEM.L\_SHIPDATE <= dateadd(dd, -90, cast(`1998-12-01' as datetime))
GROUP BY LINEITEM.L\_RETURNFLAG, LINEITEM.L\_LINESTATUS}
 &\texttt{SUM(A)}\\
$Q_3$&\texttt{SELECT TOP 10 LINEITEM.L\_ORDERKEY, SUM(LINEITEM.L\_EXTENDEDPRICE), ORDERS.O\_ORDERDATE, ORDERS.O\_SHIPPRIORITY FROM CUSTOMER, ORDERS, LINEITEM 
WHERE CUSTOMER.C\_MKTSEGMENT = `BUILDING' AND CUSTOMER.C\_CUSTKEY = ORDERS.O\_CUSTKEY AND LINEITEM.L\_ORDERKEY = ORDERS.O\_ORDERKEY AND ORDERS.O\_ORDERDATE < `1995-03-15' AND LINEITEM.L\_SHIPDATE > `1995-03-15' GROUP BY LINEITEM.L\_ORDERKEY, ORDERS.O\_ORDERDATE, ORDERS.O\_SHIPPRIORITY}
&\texttt{SUM(A)}\\
$Q_4$&\texttt{SELECT ORDERS.O\_ORDERPRIORITY, COUNT(*) AS O\_COUNT FROM ORDERS WHERE ORDERS.O\_ORDERDATE >= `1993-07-01' AND ORDERS.O\_ORDERDATE < dateadd(mm,3, cast(`1993-07-01' as datetime)) GROUP BY ORDERS.O\_ORDERPRIORITY}
&\texttt{COUNT(*)}\\
$Q_5$&\texttt{SELECT NATION.N\_NAME, SUM(LINEITEM.L\_EXTENDEDPRICE*(1-LINEITEM.L\_DISCOUNT)) AS REVENUE
FROM CUSTOMER, ORDERS, LINEITEM, SUPPLIER, NATION, REGION
WHERE CUSTOMER.C\_CUSTKEY = ORDERS.O\_CUSTKEY AND LINEITEM.L\_ORDERKEY = ORDERS.O\_ORDERKEY AND LINEITEM.L\_SUPPKEY = SUPPLIER.S\_SUPPKEY
AND CUSTOMER.C\_NATIONKEY = SUPPLIER.S\_NATIONKEY AND SUPPLIER.S\_NATIONKEY = NATION.N\_NATIONKEY AND NATION.N\_REGIONKEY = REGION.R\_REGIONKEY
AND REGION.R\_NAME = `ASIA' AND ORDERS.O\_ORDERDATE >= `1994-01-01' 
AND ORDERS.O\_ORDERDATE < DATEADD(YY, 1, cast(`1994-01-01' as datetime))
GROUP BY NATION.N\_NAME}&\texttt{SUM(A)}\\
$Q_{10}$&\texttt{SELECT TOP 20 CUSTOMER.C\_CUSTKEY, CUSTOMER.C\_NAME, SUM(LINEITEM.L\_EXTENDEDPRICE*(1-LINEITEM.L\_DISCOUNT)) AS REVENUE, CUSTOMER.C\_ACCTBAL,
NATION.N\_NAME, CUSTOMER.C\_ADDRESS, CUSTOMER.C\_PHONE
FROM CUSTOMER, ORDERS, LINEITEM, NATION
WHERE CUSTOMER.C\_CUSTKEY = ORDERS.O\_CUSTKEY AND LINEITEM.L\_ORDERKEY = ORDERS.O\_ORDERKEY AND ORDERS.O\_ORDERDATE>= `1993-10-01' AND
ORDERS.O\_ORDERDATE < dateadd(mm, 3, cast(`1993-10-01' as datetime)) AND
LINEITEM.L\_RETURNFLAG = `R' AND CUSTOMER.C\_NATIONKEY = NATION.N\_NATIONKEY
GROUP BY CUSTOMER.C\_CUSTKEY, CUSTOMER.C\_NAME, CUSTOMER.C\_ACCTBAL, CUSTOMER.C\_PHONE, NATION.N\_NAME, CUSTOMER.C\_ADDRESS}&\texttt{SUM(A)}\\
$Q_{12}$&\texttt{SELECT LINEITEM.L\_SHIPMODE,
COUNT(*) AS HIGH\_LINE\_COUNT
FROM ORDERS, LINEITEM
WHERE ORDERS.O\_ORDERKEY = LINEITEM.L\_ORDERKEY AND LINEITEM.L\_SHIPMODE IN (`MAIL',`SHIP') AND (ORDERS.O\_ORDERPRIORITY = `1-URGENT' OR ORDERS.O\_ORDERPRIORITY = `2-HIGH')
AND LINEITEM.L\_COMMITDATE < LINEITEM.L\_RECEIPTDATE AND LINEITEM.L\_SHIPDATE < LINEITEM.L\_COMMITDATE AND LINEITEM.L\_RECEIPTDATE >= `1994-01-01'
AND LINEITEM.L\_RECEIPTDATE < dateadd(mm, 1, cast(`1995-09-01' as datetime))
GROUP BY LINEITEM.L\_SHIPMODE}&\texttt{COUNT(*)}\\
\hline\end{tabular}
\label{tpch-data-queries-grouping}
\end{table}

\begin{figure}[ht]
\begin{tikzpicture}[every axis/.style={height=4.5cm,width=\linewidth,ybar stacked, xlabel= TPC-H-inspired aggregation queries with grouping, xtick=data, xticklabels={$Q_1$,$Q_3$,$Q_4$,$Q_5$,$Q_{10}$,$Q_{12}$}, legend cell align = {left}, legend columns=2, bar width=8pt, ymin=0,ymax=14, y label style={at={(0.1,0.5)}}, x label style={at={(0.5,0)}}, ylabel=Eval.\ time (seconds),legend style={draw=none, font=\scriptsize, fill=none, at={(1, 1)}}}]
\begin{axis}[bar shift=-10pt]
\addplot[fill=cyan, draw=none,  postaction={pattern=crosshatch dots, samples=20}] coordinates {(1,2.84)(2,5.418)(3,0.177)(4,0.628)(5,4.052)(6,0.737)};\label{encoding}
\addplot[fill=black, draw=none] coordinates {(1,0.077)(2,1.695)(3,0.016)(4,0.016)(5,0.455)(6,0.015)};\label{solving}
\addplot[fill=cyan!40!white, draw=none] coordinates {(1,12)(2,2.272)(3,2.08)(4,2.063)(5,1.257)(6,4.451)};\label{encoding1}
\addplot[fill=magenta, draw=none] coordinates {(1,0)(2,0.368)(3,0.678)(4,0.513)(5,0.612)(6,0.599)};\label{solving1}
\end{axis}
\begin{axis}[hide axis]
\addplot[fill=lightgray, draw=none] coordinates {(1,14)(2,6.3)(3,3.8)(4,0)(5,6.972)(6,8.237)}
node[above] at (axis cs:4,0) {$\times$}
node[rotate=90, yshift=0.02cm] at (axis cs:1,10) {\textbf{\small 2 mins $\rightarrow$}}; \label{conquer}
\end{axis}
\begin{axis}[bar shift=10pt, hide axis]
\addlegendimage{/pgfplots/refstyle=encoding}\addlegendentry{Encoding underlying CQ}
\addlegendimage{/pgfplots/refstyle=solving}\addlegendentry{Solving underlying CQ}
\addlegendimage{/pgfplots/refstyle=encoding1}\addlegendentry{Encoding for groups}
\addlegendimage{/pgfplots/refstyle=solving1}\addlegendentry{Solving for groups}
\addlegendimage{/pgfplots/refstyle=conquer}\addlegendentry{ConQuer rewriting}
\addplot[fill=red, draw=none] coordinates {(1,1.83)(2,0.827)(3,0.472)(4,1.487)(5,1.176)(6,1.026)}node[rotate=90, yshift=0.37cm] at (axis cs:1,9.65) {\textbf{\small 15 mins $\rightarrow$}};
\addlegendentry{Original query}
\end{axis}
\end{tikzpicture}
\caption{AggCAvSAT vs.\ ConQuer on TPC-H data generated using the \texttt{DBGen}-based tool (10\% inconsistency, 1 GB repairs)}
\label{w-grouping}
\end{figure}

\begin{figure}[ht]
    \centering
\begin{subfigure}[b]{0.5\textwidth}
\begin{tikzpicture}[every axis/.style={height=4.3cm,width=\linewidth,ybar stacked, xtick=data, xticklabels={$Q_1$,$Q_3$,$Q_4$,$Q_5$,$Q_{10}$,$Q_{12}$}, legend cell align = {left}, bar width=8pt, ymin=0,ymax=8, xmin=1.5, y label style={at={(0.1,0.5)}}, x label style={at={(0.5,0)}}, ylabel=Eval.\ time (seconds),legend columns = 2,legend style={draw=none, fill=none,font=\scriptsize, at={(0.72, 0.98)}}}]
\begin{axis}[bar shift=-10pt]
\addplot[fill=cyan, draw=none,  postaction={pattern=crosshatch dots, samples=20}] coordinates{(1,0)(2,1.388)(3,0.177)(4,0.527)(5,3.147)(6,0.769)};\label{encoding}
\addplot[fill=black, draw=none]coordinates{(1,0)(2,0.249)(3,0.015)(4,0.014)(5,0.368)(6,0.013)};\label{solving}
\addplot[fill=cyan!40!white, draw=none] coordinates {(1,0)(2,0.772)(3,1.432)(4,1.823)(5,0.565)(6,3.824)};\label{encoding1}
\addplot[fill=magenta, draw=none] coordinates {(1,0)(2,0.285)(3,0.212)(4,0.273)(5,0.467)(6,0.182)};\label{solving1}
\end{axis}
\begin{axis}[hide axis]
\addplot[fill=lightgray, draw=none] coordinates {(1,16)(2,3.78)(3,3.8)(4,0)(5,5.74)(6,7.8)}
node[above] at (axis cs:4,0) {$\times$}
node[rotate=90, yshift=0.02cm] at (axis cs:1,13.2) {\textbf{\small 72 secs $\rightarrow$}}node[yshift=0.37cm, draw=black] at (axis cs:5.9,6.3){\textbf{\small Instance 1}}; \label{conquer}
\end{axis}
\begin{axis}[bar shift=10pt, hide axis]
\addlegendimage{/pgfplots/refstyle=encoding}\addlegendentry{Encoding underlying CQ}
\addlegendimage{/pgfplots/refstyle=solving}\addlegendentry{Solving underlying CQ}
\addlegendimage{/pgfplots/refstyle=encoding1}\addlegendentry{Encoding for groups}
\addlegendimage{/pgfplots/refstyle=solving1}\addlegendentry{Solving for groups}
\addlegendimage{/pgfplots/refstyle=conquer}\addlegendentry{ConQuer rewriting}
\addplot[fill=red, draw=none] coordinates {(1,2.9)(2,0.627)(3,0.372)(4,1.54)(5,0.845)(6,1.2)};
\addlegendentry{Original query}
\end{axis}
\end{tikzpicture}
\end{subfigure}
\begin{subfigure}[b]{0.5\textwidth}
\begin{tikzpicture}[every axis/.style={height=4.3cm,width=\linewidth,ybar stacked, xtick=data, xticklabels={$Q_1$,$Q_3$,$Q_4$,$Q_5$,$Q_{10}$,$Q_{12}$}, legend cell align = {left}, bar width=8pt, ymin=0,ymax=8,xmin=1.5, y label style={at={(0.1,0.5)}}, x label style={at={(0.5,0)}}, ylabel=Eval.\ time (seconds),legend columns=2,legend style={draw=none,font=\scriptsize, fill=none, at={(0.72, 0.98)}}}]
\begin{axis}[bar shift=-10pt]
\addplot[fill=cyan, draw=none,  postaction={pattern=crosshatch dots, samples=20}] coordinates{(1,0)(2,1.201)(3,0.189)(4,0.706)(5,3.056)(6,0.709)};\label{encoding}
\addplot[fill=black, draw=none] coordinates{(1,0)(2,0.421)(3,0.017)(4,0.016)(5,0.423)(6,0.013)};\label{solving}
\addplot[fill=cyan!40!white, draw=none] coordinates {(1,0)(2,0.718)(3,1.302)(4,1.638)(5,0.837)(6,3.913)};\label{encoding1}
\addplot[fill=magenta, draw=none] coordinates {(1,0)(2,0.44)(3,0.458)(4,0.589)(5,0.312)(6,0.254)};\label{solving1}
\end{axis}
\begin{axis}[hide axis]
\addplot[fill=lightgray, draw=none] coordinates {(1,8)(2,3.3)(3,3.8)(4,0)(5,5.54)(6,7.57)}
node[above] at (axis cs:4,0) {$\times$}
node[rotate=90, yshift=0.02cm] at (axis cs:1,6.8) {\textbf{\small 64 secs $\rightarrow$}}node[yshift=0.37cm, draw=black] at (axis cs:5.9,6.3){\textbf{\small Instance 2}}; \label{conquer}
\end{axis}
\begin{axis}[bar shift=10pt, hide axis]
\addlegendimage{/pgfplots/refstyle=encoding}
\addlegendimage{/pgfplots/refstyle=solving}
\addlegendimage{/pgfplots/refstyle=encoding1}
\addlegendimage{/pgfplots/refstyle=solving1}
\addlegendimage{/pgfplots/refstyle=conquer}
\addplot[fill=red, draw=none] coordinates {(1,2.9)(2,0.527)(3,0.372)(4,1.54)(5,0.845)(6,1.026)};
\end{axis}
\end{tikzpicture}
\end{subfigure}
\begin{subfigure}[b]{0.5\textwidth}
\begin{tikzpicture}[every axis/.style={height=4.3cm,width=\linewidth,ybar stacked, xtick=data, xticklabels={$Q_1$,$Q_3$,$Q_4$,$Q_5$,$Q_{10}$,$Q_{12}$}, legend cell align = {left}, bar width=8pt, ymin=0,ymax=12, xmin=1.5, y label style={at={(0.1,0.5)}}, x label style={at={(0.5,0)}}, ylabel=Eval.\ time (seconds),legend columns=2,legend style={draw=none, font=\scriptsize, fill=none, at={(0.72, 0.98)}}}]
\begin{axis}[bar shift=-10pt]
\addplot[fill=cyan, draw=none,  postaction={pattern=crosshatch dots, samples=20}] coordinates{(1,0)(2,3.346)(3,0.284)(4,0.631)(5,6.189)(6,0.796)};\label{encoding}
\addplot[fill=black, draw=none]coordinates{(1,0)(2,0.948)(3,0.015)(4,0.009)(5,1.635)(6,0.013)};\label{solving}
\addplot[fill=cyan!40!white, draw=none] coordinates {(1,0)(2,0.933)(3,2.219)(4,2.936)(5,2.851)(6,4.332)};\label{encoding1}
\addplot[fill=magenta, draw=none] coordinates {(1,0)(2,0.389)(3,0.488)(4,0.768)(5,0.346)(6,0.813)};\label{solving1}
\end{axis}
\begin{axis}[hide axis]
\addplot[fill=lightgray, draw=none] coordinates {(1,13)(2,3.3)(3,3.8)(4,0)(5,5.54)(6,7.57)}
node[above] at (axis cs:4,0) {$\times$}
node[rotate=90, yshift=0.02cm] at (axis cs:1,11.2) {\textbf{\small 68 secs $\rightarrow$}}node[yshift=0.37cm, draw=black] at (axis cs:5.9,9.4){\textbf{\small Instance 3}}; \label{conquer}
\end{axis}
\begin{axis}[bar shift=10pt, hide axis]
\addlegendimage{/pgfplots/refstyle=encoding}
\addlegendimage{/pgfplots/refstyle=solving}
\addlegendimage{/pgfplots/refstyle=encoding1}
\addlegendimage{/pgfplots/refstyle=solving1}
\addlegendimage{/pgfplots/refstyle=conquer}
\addplot[fill=red, draw=none] coordinates {(1,2.9)(2,0.527)(3,0.372)(4,1.54)(5,0.845)(6,1.026)};
\end{axis}
\end{tikzpicture}
\end{subfigure}
\begin{subfigure}[b]{0.5\textwidth}
\begin{tikzpicture}[every axis/.style={height=4.3cm,width=\linewidth,ybar stacked, xtick=data,xlabel= TPC-H-inspired aggregation queries with grouping, xticklabels={$Q_1$,$Q_3$,$Q_4$,$Q_5$,$Q_{10}$,$Q_{12}$}, legend cell align = {left}, bar width=8pt, ymin=0,ymax=15,xmin=1.5, y label style={at={(0.1,0.5)}}, x label style={at={(0.5,0)}}, ylabel=Eval.\ time (seconds),legend columns=2,legend style={draw=none, font=\scriptsize, fill=none, at={(0.72, 0.98)}}}]
\begin{axis}[bar shift=-10pt]
\addplot[fill=cyan, draw=none,  postaction={pattern=crosshatch dots, samples=20}] coordinates{(1,0)(2,6.0256)(3,0.26)(4,0.597)(5,6.763)(6,0.763)};\label{encoding}
\addplot[fill=black, draw=none]coordinates{(1,0)(2,1.456)(3,0.016)(4,0.013)(5,4.558)(6,0.015)};\label{solving}
\addplot[fill=cyan!40!white, draw=none] coordinates {(1,0)(2,2.172)(3,2.686)(4,3.298)(5,2.151)(6,2.525)};\label{encoding1}
\addplot[fill=magenta, draw=none] coordinates {(1,0)(2,0.333)(3,0.563)(4,1.913)(5,0.323)(6,1.3251)};\label{solving1}
\end{axis}
\begin{axis}[hide axis]
\addplot[fill=lightgray, draw=none] coordinates {(1,16)(2,3.78)(3,3.8)(4,0)(5,5.74)(6,7.8)}
node[above] at (axis cs:4,0) {$\times$}
node[rotate=90, yshift=0.02cm] at (axis cs:1,13.2) {\textbf{\small 72 secs $\rightarrow$}}node[yshift=0.37cm, draw=black] at (axis cs:5.9,11.8){\textbf{\small Instance 4}}; \label{conquer}
\end{axis}
\begin{axis}[bar shift=10pt, hide axis]
\addlegendimage{/pgfplots/refstyle=encoding}
\addlegendimage{/pgfplots/refstyle=solving}
\addlegendimage{/pgfplots/refstyle=encoding1}
\addlegendimage{/pgfplots/refstyle=solving1}
\addlegendimage{/pgfplots/refstyle=conquer}
\addplot[fill=red, draw=none] coordinates {(1,2.9)(2,0.627)(3,0.372)(4,1.54)(5,0.845)(6,1.2)};
\end{axis}
\end{tikzpicture}
\end{subfigure}
\caption{AggCAvSAT vs.\ ConQuer on PDBench instances}
\label{w-grouping-pdbench}
\end{figure}

The overhead of computing the range consistent answers to aggregation queries with grouping is higher than that for the aggregation queries without grouping because for an aggregation query with grouping, AggCAvSAT needs to construct and solve twice as many \wpmaxsat{} instances as there are consistent groups, i.e., one for the \textit{lub}-answer and one for the \textit{glb}-answer per consistent group. For queries that involved the \texttt{SELECT TOP $k$} construct of SQL, we chose top $k$ consistent groups ordered by one or more grouping attributes present in the \texttt{ORDER BY} clause of the query. AggCAvSAT computes the range consistent answers to each query under ten seconds except for $Q_1$. It took under three seconds to compute the consistent groups of $Q_1$, but took over forty seconds to encode the range consistent answers of the groups and over fifteen minutes to solve the corresponding \wpmaxsat{} instances. This is because some consistent groups have over 3M tuples and so the \wpmaxsat{} instances have over 600 thousand variables and over 1.3 million clauses. ConQuer took slightly over two minutes to compute the range consistent answers to $Q_1$. We did not include $Q_1$ in experiments with larger databases and higher inconsistency.

Figure \ref{w-grouping-pdbench} shows the comparison of AggCAvSAT and ConQuer for aggregation queries with grouping on PDBench instances. For the database instance with the lowest amount of inconsistency, AggCAvSAT beats ConQuer on all queries, but as the inconsistency grows, AggCAvSAT takes longer time to encode and solve for the consistent groups of the queries $Q_3$ and $Q_{10}$.

In Figure \ref{varying-inconsistency-grouping}, we first plot the evaluation time of AggCAvSAT as the percentage of inconsistency in the data grows from 5\% to 35\% in the instances generated using the \texttt{DBGen}-based data generator. The size of the database repairs is kept constant at 1 GB (8 million tuples).
\begin{figure}[ht]
\begin{subfigure}[b]{0.5\textwidth}
\begin{tikzpicture}
  \begin{axis}[width=\linewidth, grid=major,height=4.1cm, tick label style={font=\normalsize},legend columns=5, label style={font=\normalsize}, grid style={white}, ylabel={Eval.\ time (seconds)}, y label style={at={(0.05,0.5)}}, ymax = 50, legend style={draw=none, fill=none, at={(0.9,1)}},legend cell align={left}, every axis plot/.append style={thick}]

  \ignore{\addplot[blue, mark=*] coordinates{(5,1.664) (10,2.074) (15,1.643) (20,1.7) (25,1.799)(30,1.964) (35,1.878)};\addlegendentry{$Q_1$}}
  \addplot[cyan, mark=triangle] coordinates{(5,4.775) (10,9.75) (15,14.337) (20,19.003) (25,26.387)(30,33.662) (35,41.048)};\addlegendentry{$Q_3$}
  \addplot[orange, mark=o] coordinates{(5,2.565) (10,2.95) (15,3.042) (20,4.284) (25,5.108)(30,6.473) (35,8.043)};\addlegendentry{$Q_4$}
  \addplot[black, mark=x] coordinates{(5,2.841) (10,3.218) (15,4.198) (20,5.312) (25,7.163)(30,11.3) (35,16.5)};\addlegendentry{$Q_5$}
  \addplot[purple, mark=square] coordinates{(5,4.774) (10,6.376) (15,9.192)(20,14.506)(25,21.112)(30,31.047)(35,44.027)};\addlegendentry{$Q_{10}$}
  \addplot[violet, mark=star] coordinates{(5,5.017) (10,5.801) (15,6.296) (20,7.39) (25,9.558)(30,13.419) (35,15.827)};\addlegendentry{$Q_{12}$}
\end{axis}
\end{tikzpicture}
\end{subfigure}
\begin{subfigure}[b]{0.5\textwidth}
\begin{tikzpicture}
  \begin{axis}[width=\linewidth, grid=major,height=4.1cm, tick label style={font=\normalsize},legend columns=5, label style={font=\normalsize}, grid style={white}, xlabel={Percentage of inconsistency}, ylabel={Number of SAT calls}, y label style={at={(0.05,0.5)}}, y tick label style={/pgf/number format/.cd, sci, sci zerofill, precision=0}, legend style={draw=none, fill=none,at={(0.9,1)}},legend cell align={left}, ymode=log, ymax=50000, every axis plot/.append style={thick}]
  \ignore{\addplot[blue, mark=*] coordinates{(5,1.664) (10,2.074) (15,1.643) (20,1.895) (25,1.799)(30,1.964) (35,1.878)};\addlegendentry{$Q_1$}}
  \addplot[cyan, mark=triangle] coordinates{(5,23) (10,29) (15,35) (20,34) (25,36)(30,52) (35,64)};\addlegendentry{$Q_3$}
  \addplot[orange, mark=o] coordinates{(5,19) (10,19) (15,31) (20,34) (25,41)(30,40) (35,41)};\addlegendentry{$Q_4$}
  \addplot[black, mark=x] coordinates{(5,71) (10,139) (15,364) (20,712) (25,1061)(30,2238) (35,4254)};\addlegendentry{$Q_5$}
   \addplot[purple, mark=square] coordinates{(5,416) (10,767) (15,1190) (20,1613) (25,2128)(30,2599) (35,3063)};\addlegendentry{$Q_{10}$}
	\addplot[violet, mark=star] coordinates{(5,94) (10,153) (15,490) (20,899) (25,1225)(30,1934) (35,2306)};\addlegendentry{$Q_{12}$}
\end{axis}
\end{tikzpicture}
\end{subfigure}
\caption{AggCAvSAT on TPC-H data generated using the \texttt{DBGen}-based tool (varying inconsistency, 1 GB repairs)}
\label{varying-inconsistency-grouping}
\end{figure}
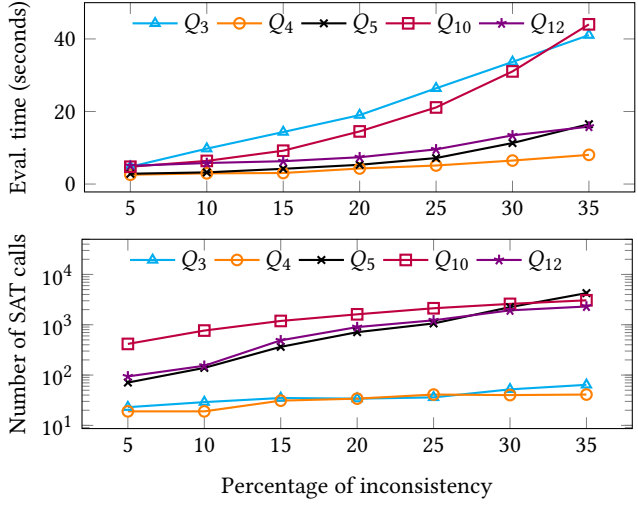
Since AggCAvSAT constructs and solves many \wpmaxsat{} instances having varying sizes for an aggregation query involving grouping, we also plot the overall number of SAT calls made by the solver in Figure \ref{varying-inconsistency-grouping}. Note that the Y-axis has logarithmic scaling in the second plot of Figure \ref{varying-inconsistency-grouping}. There are ten consistent groups in the answers to $Q_3$, and just five and two consistent groups in the answers to $Q_5$ and $Q_{12}$ respectively. In each consistent group, the aggregation operator is applied over a much larger set of tuples in $Q_5$ and $Q_{12}$ than in $Q_3$. As a result, the evaluation time for $Q_3$ is high but the number of SAT calls is comparatively less, while AggCAvSAT makes more SAT calls for $Q_5$ and $Q_{12}$, even though their consistent answers are computed much faster. The query $Q_{10}$ requires long time to construct and solve the \wpmaxsat{} instances for its consistent groups due to its high selectivity and the presence of joins between four relations. The evaluation time of computing the range consistent answers to aggregation queries with grouping increases almost linearly w.r.t.\ the size of the database when the percentage of inconsistency is constant (Figure \ref{varying-datasize-grouping}). The second plot in Figure \ref{varying-datasize-grouping}
depicts the number of SAT calls made by the solver as the size of the database grows. Due to low selectivity, the answers to $Q_4$ are encoded into small CNF formulas even on databases with high inconsistency or large sizes, resulting in fast evaluations.

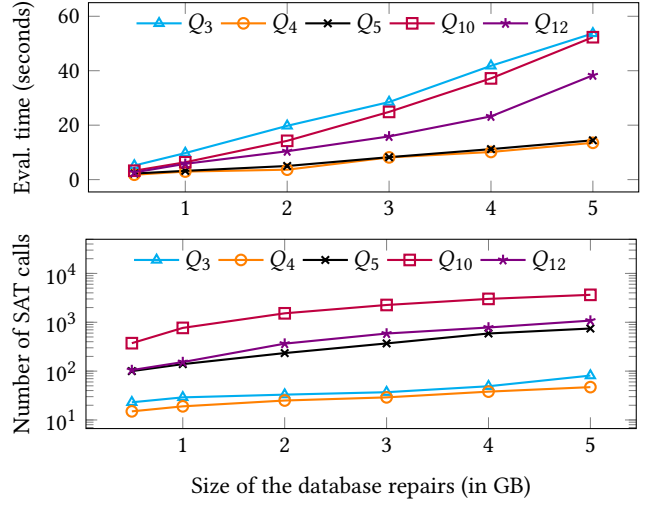
\begin{figure}[ht]
\begin{subfigure}[b]{0.5\textwidth}
\begin{tikzpicture}
  \begin{axis}[width=\linewidth, grid=major, height=4.1cm, tick label style={font=\normalsize},legend columns=5, label style={font=\normalsize}, grid style={white}, ylabel={Eval.\ time (seconds)}, y label style={at={(0.05,0.5)}}, ymax=65, y tick label style={/pgf/number format/.cd, fixed, fixed zerofill, precision=0}, legend style={draw=none, fill=none,at={(0.9,1)}},legend cell align={left}, every axis plot/.append style={thick}]
 
  \ignore{\addplot[blue, mark=*] coordinates{(0.5,0.898)(1,2.074)(2,0)(3,5.109)(4,7.942)(5,12.896)};\addlegendentry{$Q_1$}}
  \addplot[cyan, mark=triangle] coordinates{(0.5,5.177)(1,9.75)(2,19.754)(3,28.489)(4,41.797)(5,53.698)};\addlegendentry{$Q_3$}
  \addplot[orange, mark=o] coordinates{(0.5,1.792)(1,2.95)(2,3.648)(3,8.158)(4,10.15)(5,13.513)};\addlegendentry{$Q_4$}
  \addplot[black, mark=x] coordinates{(0.5,2.235)(1,3.218)(2,5.004)(3,8.243)(4,11.199)(5,14.442)};\addlegendentry{$Q_5$}
  \addplot[purple, mark=square] coordinates{(0.5,3.262)(1,6.376)(2,14.246)(3,24.872)(4,37.18)(5,52.309)};\addlegendentry{$Q_{10}$}
  \addplot[violet, mark=star] coordinates{(0.5,2.646)(1,5.801)(2,10.421)(3,15.87)(4,23.232)(5,38.312)};\addlegendentry{$Q_{12}$}
\end{axis}
\end{tikzpicture}
\end{subfigure}
\begin{subfigure}[b]{0.5\textwidth}
\begin{tikzpicture}
  \begin{axis}[width=\linewidth, grid=major,height=4.1cm, tick label style={font=\normalsize},legend columns=5, label style={font=\normalsize}, grid style={white}, xlabel={Size of the database repairs (in GB)}, ylabel={Number of SAT calls}, y label style={at={(0.05,0.5)}}, y tick label style={/pgf/number format/.cd, sci, sci zerofill, precision=0}, legend style={draw=none, fill=none,at={(0.9,1)}},legend cell align={left}, ymode=log, ymax=50000, every axis plot/.append style={thick}]
  \ignore{\addplot[blue, mark=*] coordinates{(0.5,0.898)(1,2.074)(2,0)(3,5.109)(4,7.942)(5,12.896)};\addlegendentry{$Q_1$}}
  \addplot[cyan, mark=triangle] coordinates{(0.5,23)(1,29)(2,33)(3,37)(4,49)(5,81)};\addlegendentry{$Q_3$}
  \addplot[orange, mark=o] coordinates{(0.5,15)(1,19)(2,25)(3,29)(4,38)(5,47)};\addlegendentry{$Q_4$}
  \addplot[black, mark=x] coordinates{(0.5,100)(1,139)(2,234)(3,369)(4,589)(5,745)};\addlegendentry{$Q_5$}
  \addplot[purple, mark=square] coordinates{(0.5,373)(1,767)(2,1525)(3,2259)(4,3011)(5,3645)};\addlegendentry{$Q_{10}$}
  \addplot[violet, mark=star] coordinates{(0.5,105)(1,153)(2,365)(3,588)(4,782)(5,1081)};\addlegendentry{$Q_{12}$}
\end{axis}
\end{tikzpicture}
\end{subfigure}
\caption{AggCAvSAT on TPC-H data generated using the \texttt{DBGen}-based tool (varying database sizes, 10\% inconsistency)}
\label{varying-datasize-grouping}
\end{figure}

\subsubsection{Discussion}
The experiments show that AggCAvSAT performed well across a broad range of queries and databases; it performed worse  on queries with high selectivity because, in such cases, very large CNF formulas were generated. AggCAvSAT slowed down on databases with high degree of inconsistency ($> 30\%$) and with key-equal groups of large sizes ($>15$). These are rather corner cases that should not be encountered  in real-world databases.

\subsection{Experiments with Real-world Data}
\subsubsection{Dataset} For this set experiments, we use the schema and the data from Medigap \cite{medigap}, an openly available real-world database about Medicare supplement insurance in the United States. We combine the data from 2019 and 2020 to obtain a database with over 61K tuples (Table \ref{real-world-data-schema}).  We evaluated the performance of Reduction \ref{reduction3}, since we consider two functional dependencies and one denial constraint on the Medigap schema, as shown in Table \ref{real-world-integrity-constraints}. The actual data was inconsistent so no  additional inconsistency was injected.
\begin{table}[ht]
\caption{Medigap real-world database}
\begin{subtable}{0.45\textwidth}
\centering
\begin{tabular}{|l|l|c|c|} \hline
Relation&Acronym&\# of attributes&\# of tuples\\ \hline
OrgsByState &OBS& 5 & 3872\\
PlansByState &PBS& 18 & 21002\\
PlansByZip &PBZ& 20 & 4748\\
PlanType &PT& 4 & 2434\\
Premiums &PR& 7 & 29148\\
SimplePlanType &SPT& 4 & 70\\
\hline\end{tabular}
\caption{Medigap schema}
\label{real-world-data-schema}
\end{subtable}%

\smallskip

\begin{subtable}{0.45\textwidth}
\centering
\begin{tabular}{|l|l|c|} \hline
Type&Constraint Definition&Inconsistency\\ \hline
FD & OBS (orgID $\rightarrow$ orgName)&2.58\%\\
FD & PBS (addr, city, abbrev $\rightarrow$ zip)&1.5\%\\
DC & $\forall\; t \in \text{PBS}$ ($t$.webAddr $\neq$ `')&0.15\%\\
\hline\end{tabular}
\caption{Integrity constraints and inconsistency}
\label{real-world-integrity-constraints}
\end{subtable}%
\label{real-data}
\end{table}

\vspace{-0.1cm}

\subsubsection{Queries}
We use twelve natural aggregation queries on the Medigap database that involve the aggregation operators \texttt{COUNT(*)}, \texttt{COUNT($A$)}, and \texttt{SUM($A$)}. We refer to these as $(Q^m_1, \cdots, Q^m_{12})$. The first six queries 
contain no grouping, while the rest of them do. The definitions of these queries are given in Table \ref{real-data-queries}.

\begin{table}[h]
\scriptsize
\caption{Aggregation Queries on Medigap database}
\centering
\begin{tabular}{|p{0.03\linewidth}|p{0.9\linewidth}|} \hline
\#&Query\\ \hline
$Q^m_1$&\texttt{SELECT COUNT(*) FROM OBS WHERE OBS.Name = `Continental General Insurance Company'}\\
$Q^m_2$&\texttt{SELECT COUNT(*) FROM PBZ, SPT WHERE PBZ.Description = SPT.Simple\_plantype\_name AND SPT.Contract\_year = 2020 AND SPT.Simple\_plantype = `B'}\\
$Q^m_3$&\texttt{SELECT SUM(PBZ.Over65) FROM PBZ WHERE PBZ.State\_name = `Wisconsin' AND PBZ.County\_name = `GREEN LAKE'}\\
$Q^m_4$&\texttt{SELECT SUM(PBZ.Community) FROM PBZ WHERE PBZ.State\_name = `New York'}\\
$Q^m_5$&\texttt{SELECT COUNT(PR.Premium\_range) FROM PR}\\
$Q^m_6$&\texttt{SELECT COUNT(PR.Premium\_range) FROM PT, PR WHERE PT.State\_abbrev = PR.State\_abbrev AND PT.Plan\_type = PR.Plan\_type AND PT.Contract\_year = PR.Contract\_year AND PT.Contract\_year = 2020 AND PT.Simple\_plantype = `K'}\\
$Q^m_7$&\texttt{SELECT SPT.Contract\_year, COUNT(*) FROM SPT GROUP BY SPT.Contract\_year ORDER BY SPT.Contract\_year DESC}\\
$Q^m_8$&\texttt{SELECT PBZ.State\_name, COUNT(*) FROM PBZ GROUP BY PBZ.State\_name
}\\
$Q^m_9$&\texttt{SELECT PBZ.Zip, SUM(PBZ.Community) FROM PBZ WHERE PBZ.State\_name = `New York' GROUP BY PBZ.Zip}\\
$Q^m_{10}$&\texttt{SELECT TOP 10 PBS.State\_name, SPT.Contract\_year, SUM(PBS.Under65) FROM PBS, SPT WHERE SPT.Simple\_plantype\_name = PBS.Description AND SPT.Simple\_plantype = `A' AND SPT.Language\_id = 1 GROUP BY PBS.State\_name, SPT.Contract\_year ORDER BY PBS.State\_name}\\
$Q^m_{11}$&\texttt{SELECT PR.Age\_category, COUNT(PR.Premium\_range) FROM PR GROUP BY PR.Age\_category ORDER BY PR.Age\_category}\\
$Q^m_{12}$&\texttt{SELECT TOP 10 PT.Simple\_plantype, COUNT(PR.Premium\_range) FROM PT, PR WHERE PT.State\_abbrev = PR.State\_abbrev AND PT.Plan\_type = PR.Plan\_type AND PT.Contract\_year = PR.Contract\_year and PT.Contract\_year = 2020 GROUP BY PT.Simple\_plantype ORDER BY PT.Simple\_plantype}\\
\hline\end{tabular}
\label{real-data-queries}
\end{table}

\subsubsection{Results on Real-world Database}\label{sec-exp-real-world}
\begin{figure}[ht]
\begin{tikzpicture}[every axis/.style={ybar stacked, xlabel= Real-world aggregation queries on Medigap database, xtick=data, height=4.5cm,width=\linewidth, xticklabels={$Q^m_1$,$Q^m_2$,$Q^m_3$,$Q^m_4$,$Q^m_5$,$Q^m_6$,$Q^m_7$,$Q^m_8$,$Q^m_9$,$Q^m_{10}$,$Q^m_{11}$,$Q^m_{12}$}, legend cell align = {left}, bar width=8pt, ymin=0,ymax=45, y label style={at={(0.1,0.5)}}, x label style={at={(0.5,0)}}, ylabel=Eval.\ time (seconds),legend style={draw=none, fill=none,font=\scriptsize, at={(0.5, 0.95)}}}]
\begin{axis}
\addplot[fill=cyan, draw=none,postaction={pattern=crosshatch dots, samples=20}] coordinates {(1,9.868)(2,11.243)(3,13.334)(4,12.696)(5,13.618)(6,14.218)(7,12.306)(8,12.603)(9,13.271)(10,14.123)(11,12.769)(12,13.427)};\label{encoding}
\addplot[fill=black, draw=none] coordinates {(1,0.495)(2,0.44)(3,0.427)(4,0.452)(5,0.771)(6,0.544)(7,0.326)(8,0.255)(9,0.229)(10,0.292)(11,0.311)(12,0.283)};\label{solving}
\addplot[fill=cyan!40!white, draw=none] coordinates {(1,0)(2,0)(3,0)(4,0)(5,0)(6,0)(7,19.555)(8,14.327)(9,14.402)(10,25.515)(11,20.393)(12,25.818)};\label{encoding1}
\addplot[fill=magenta, draw=none] coordinates {(1,0)(2,0)(3,0)(4,0)(5,0)(6,0)(7,0.723)(8,1.139)(9,0.537)(10,3.488)(11,2.204)(12,3.485)};\label{solving1}
\end{axis}
\begin{axis}[bar shift=10pt, hide axis]
\addlegendimage{/pgfplots/refstyle=encoding}\addlegendentry{Encoding underlying CQ}
\addlegendimage{/pgfplots/refstyle=solving}\addlegendentry{Solving underlying CQ}
\addlegendimage{/pgfplots/refstyle=encoding1}\addlegendentry{Encoding for groups}
\addlegendimage{/pgfplots/refstyle=solving1}\addlegendentry{Solving for groups}
\addplot[fill=red, draw=none] coordinates {(1,0)};
\end{axis}
\end{tikzpicture}
\caption{Evaluation time for computing the range consistent answers to real-world aggregation queries}
\label{real-data-evaluation}
\end{figure}
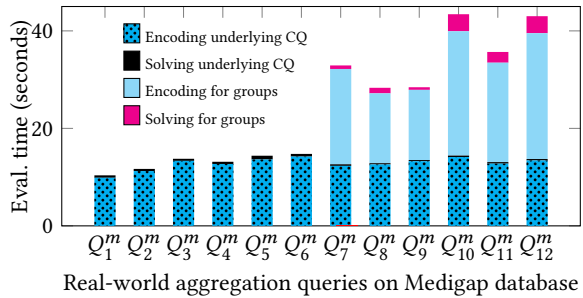

\vspace{-0.1cm}

In Figure \ref{real-data-evaluation}, we plot the overall time taken by AggCAvSAT to compute the range consistent answers to the twelve aggregation queries on the Medigap database.
Since the Medigap schema has functional dependencies and a denial constraint, the encoding of CQA into \wpmaxsat{} instances is based on Reduction \ref{reduction3}. Consequently, the size of the CNF formulas is much larger compared to that of the ones produced by Reduction \ref{reduction1}, resulting in longer encoding times. For all twelve queries, the encoding time is dominated by the time required to compute the near-violations and hence the $\gamma$-clauses. This part of the encoding time is equal for all queries, but the computation time for the witnesses depends on the query. The solver takes comparatively minuscule amount of time to compute the consistent answers to the underlying conjunctive query. For the queries $Q^m_7, \cdots, Q^m_{12}$, the \textit{glb}-answer and the \textit{lub}-answer are encoded and then solved for for each consistent group, causing high overhead. The longest evaluation time is taken by queries $Q^m_{10}$, $Q^m_{12}$, and $Q^m_6$ since they consist of 10, 10, and 6 consistent groups, respectively.

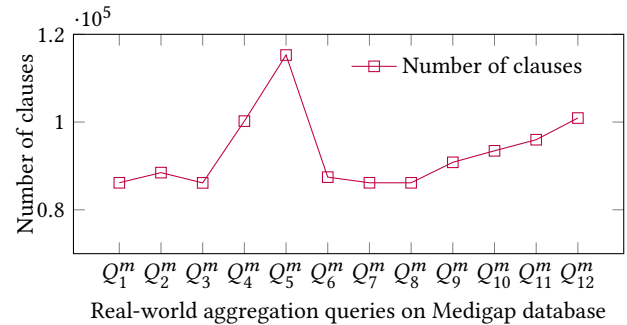
\begin{figure}[ht]
\begin{tikzpicture}[every axis/.style={xlabel= Real-world aggregation queries on Medigap database, xtick=data,height=4.5cm, width=0.5\textwidth, xticklabels={$Q^m_1$,$Q^m_2$,$Q^m_3$,$Q^m_4$,$Q^m_5$,$Q^m_6$,$Q^m_7$,$Q^m_8$,$Q^m_9$,$Q^m_{10}$,$Q^m_{11}$,$Q^m_{12}$}, legend cell align = {left}, ymin=0, y label style={at={(0.08,0.5)}}, x label style={at={(0.5,0)}},ymin=70000,ymax=120000, ylabel=Number of clauses,legend style={draw=none, fill=none, at={(0.95, 0.95)}}}]
\begin{axis}
\addplot[purple, mark=square] coordinates {(1,86156)(2,88474)(3,86135)(4,100214)(5,115274)(6,87446)(7,86160)(8,86154)(9,90823)(10,93448)(11,95978)(12,100915)};\label{clauses}
\addlegendentry{Number of clauses}
\end{axis}
\end{tikzpicture}
\caption{Number of clauses in a CNF formula capturing the consistent answers to the underlying conjunctive query}
\label{real-data-clauses}
\end{figure}

For these experiments, we did not compute the range consistent answers from the consistent part of the data first. Thus, for all CNF formulas, the number of variables for  is equal to the number of tuples in the data (about 61K). The number of clauses, however, varies depending on the query, as shown in Figure \ref{real-data-clauses}. The query $Q^m_5$ has the highest number of clauses since all tuples in the vwPremiums table are the minimal witnesses to its underlying conjunctive query.

The clauses arising from the inconsistency in the data can be constructed independently from the clauses arising from the witnesses to the queries. In the near future, we plan to parallelize their computation to improve AggCAvSAT's performance.

\section{Concluding Remarks}
First, we showed that computing the range consistent answers to an aggregation query involving the $\texttt{SUM(A)}$ operator can be NP-hard, even if the consistent answers to the underlying conjunctive query are SQL-rewritable.
We 
then 
designed, implemented, and evaluated AggCAvSAT, a SAT-based system for computing range consistent answers to aggregation queries involving 
$\texttt{COUNT(A)}$, $\texttt{COUNT(*)}$, $\texttt{SUM(A)}$, and grouping.
 It is the first system able to handle aggregation queries whose range consistent answers are not SQL-rewritable. Our experimental evaluation showed that AggCAvSAT is not only competitive with  systems such as ConQuer but it is also scalable. The experiments on the Medigap data showed that AggCAvSAT can handle real-world databases having integrity constraints beyond primary keys. The next step in this investigation is to first delineate the complexity of the range consistent answers to aggregation queries with the operator \texttt{\small AVG($A$)} and then enhance the capabilities of AggCAvSAT to compute the range consistent answers of such aggregation queries. Finally, we note that the SAT-based methods used here are applicable to broader classes of SQL queries, such as queries with nested subqueries, as long as denial constraints are considered. If broader classes of constraints are considered, such as universal constraints, then  the consistent answers of even conjunctive queries become $\Pi_2^p$-hard to compute \cite{DBLP:conf/icdt/ArmingPS16}, hence SAT-based methods are not applicable. In that case, Answer Set Programming solvers (for example, DLV \cite{DBLP:journals/tocl/LeonePFEGPS06} or Potassco \cite{DBLP:journals/aicom/GebserKKOSS11}) have to be used, instead of SAT solvers.
 \ignore{
 Finally, we note that for nested aggregation queries, the complexity of the range semantics may jump to levels of the polynomial hierarchy PH that are higher than NP. For such queries,  Answer Set Programming solvers (for example, DLV \cite{DBLP:journals/tocl/LeonePFEGPS06} or Potassco \cite{DBLP:journals/aicom/GebserKKOSS11}) have to be used, instead of SAT solvers.}

\balance


\smallskip

\noindent{\bf Acknowledgments}~
Dixit was supported by a Baskin School of Engineering Dissertation-Year Fellowship and by the Center for Research in Open Source Software (CROSS) at the UC Santa Cruz. Kolaitis was partially supported by NSF Award IIS: 1814152.
\bibliographystyle{ACM-Reference-Format}
\bibliography{arXiv-version}
\end{document}